\documentclass[11pt]{article} 
\usepackage{authblk}
\usepackage{url}
\usepackage{hyperref}
\usepackage{verbatim}
\usepackage{makeidx}
\usepackage{graphicx}
\usepackage{amssymb}
\usepackage{amsmath}
\usepackage{latexsym}
\usepackage{amsbsy}
\usepackage{epsfig}

\newcommand{\citep}{\cite}

\newtheorem{theorem}{Theorem}

\newtheorem{lemma}{Lemma}

\newtheorem{proposition}{Proposition}

\newtheorem{definition}{Definition}

\newcommand{\commentout}[1]{}
\newcommand{\prs}{\par}

\newenvironment{proof}{\noindent{\bf Proof}:}{\mbox{$\Box$}
\prs\vspace{3mm}\prs}

\newcounter{myenumeratecounter}

\newtheorem{examplehidden}{Example}
\newenvironment{example}{
\begin{examplehidden}
\em
}
{\end{examplehidden}}

\newcommand{\pval}{\ensuremath{{\tt pval}}}

\newcommand{\ecol}{\ensuremath{\textsc{\rm S}}}

\newcommand{\cA}{\ensuremath{\mathcal A}}
\newcommand{\cB}{\ensuremath{\mathcal B}}

\newcommand{\cE}{\ensuremath{\mathcal E}}
\newcommand{\cF}{\ensuremath{\mathcal F}}

\newcommand{\cH}{\ensuremath{\mathcal H}}
\newcommand{\cI}{\ensuremath{\mathcal I}}

\newcommand{\cP}{\ensuremath{\mathcal P}}

\newcommand{\cY}{\ensuremath{\mathcal Y}}

\newcommand{\reals}{{\mathbb R}}

\newcommand{\loss}{\ensuremath{\text{\sc loss}}}

\usepackage{stmaryrd}
\newcommand{\pnasfont}{\sffamily\small\bfseries}
\newcommand{\bloss}[3]{\ensuremath{L_{#1}(#2,#3)}}
\renewcommand{\loss}[2]{\ensuremath{L(#1,#2)}}
\newcommand{\snull}{\ensuremath{\underline{ 0}}}
\newcommand{\alt}{\ensuremath{\underline{ 1}}}
\newcommand{\Hnull}{\ensuremath{\cH(\snull)}}
\newcommand{\Halt}{\ensuremath{\cH(\alt)}}

\newcommand{\mean}{\ensuremath{\theta}}
\newcommand{\meanplus}{\ensuremath{\mean^{+}}}
\newcommand{\meanmin}{\ensuremath{\mean^{-}}}

\newcommand{\cd}{\ensuremath{\text{\sc cd}}}

\newcommand{\Ltheta}{\ensuremath{\theta_{\textsc{L}}}}
\newcommand{\Rtheta}{\ensuremath{\theta_{\textsc{R}}}}

\newcommand{\cs}{\ensuremath{\text{\sc cs}}}

\newcommand{\lr}{\ensuremath{\text{\sc lr}}}
\newcommand{\np}{\ensuremath{\text{\sc np}}}
\newcommand{\Snp}[1]{\ensuremath{S^{\textsc{np}(#1)}}}

\newcommand{\Slr}{\ensuremath{S^{\textsc{lr}}}}
\newcommand{\Rmean}{\ensuremath{\theta_{\textsc{R}}}}
\newcommand{\Lmean}{\ensuremath{\theta_{\textsc{L}}}}

\renewcommand{\pval}{\ensuremath{\textsc{p}}}

\begin{document}

\title{Beyond Neyman-Pearson: \\ e-values enable hypothesis testing with a data-driven alpha}
\author[12]{Peter Gr\"unwald}
\affil[1]{Centrum Wiskunde \& Informatica, Amsterdam, The Netherlands}
\affil[2]{Leiden University, Leiden, The Netherlands}
\date{\today}




\maketitle
\begin{abstract}
 A standard practice in statistical hypothesis testing is to mention the p-value alongside the accept/reject decision. We show the advantages of mentioning an e-value instead. With p-values, it is not clear how to use an extreme observation (e.g. $\pval \ll \alpha$) for getting better frequentist decisions. With e-values it is straightforward, since they provide  Type-I risk control in a generalized Neyman-Pearson setting with the decision task (a general loss function) determined post-hoc, after observation of the data --- thereby providing a handle on `roving $\alpha$'s'. When Type-II risks are taken into consideration, the only admissible decision rules in the post-hoc setting turn out to be e-value-based. Similarly, if the loss incurred when specifying a faulty confidence interval is not fixed in advance, standard confidence intervals and distributions may fail whereas e-confidence sets and e-posteriors still provide valid risk guarantees. Sufficiently powerful e-values have by now been developed for a range of classical testing problems. We discuss the main challenges for wider development and deployment. 
\end{abstract}

We 
perform a null hypothesis test with significance level $\alpha$ and we observe a p-value $\pval \ll \alpha$. Why aren't we allowed to say ``we have rejected the null at level $\pval$''? While a continuous source of bewilderment to the applied scientist, professional statisticians understand the reason: to get  a Type-I error probability guarantee of $\alpha$ --- a cornerstone of the Neyman-Pearson (NP) theory of testing --- we must set $\alpha$ in advance. But this immediately raises another question: why should the p-value  be mentioned at all in scientific papers, next to the reject/accept decision for the pre-specified $\alpha$ \citep{Berger03,Hubbard04}?
The prevailing attitude is to accept this standard practice, on the grounds that  it ``provides more information'' --- as explicitly stated by, for example, Lehmann  \cite{Lehmann93}, one of NP theory's main contributors. But this is  problematic: there is nothing in NP theory to tell us what the decision-theoretic consequences of `$\pval \ll \alpha$' could be, whereas at the same time,  the fundamental motivation behind NP theory {\em is\/} decision-theoretic: according to \cite{Neyman50B},  ``{\em {\rm [all of]} mathematical statistics deals with problems relating to performance characteristics of rules of inductive behavior {\rm [i.e. decision rules]} based on random experiments}''. 
There is no simple way though to translate observation of a $\pval$ with $\pval \ll \alpha$ into better decisions: as is well-known and reviewed below (\eqref{eq:pvalb} and \eqref{eq:calibrate}), intuitive and common decision-theoretic interpretations of $\pval \ll \alpha$ are usually just wrong. We are therefore faced with a standard practice in NP testing that, according to (strict, behaviorist) NP theory, is not part of mathematical statistics! 
\commentout{
This conundrum is sometimes  by stating that $\pval$  is a meaningful measure of `evidence against the null', even if it has no decision-theoretic consequences. In light of the developments below, this seems a rather   is highly problematic --- as has been forcefully argued by many, p-values have properties that are at odds with any reasonable definition of `evidence'; see e.g. \citep{royall1997statistical} and the many references therein.}
\paragraph{E as an alternative for P} 
In our main result, Theorem~\ref{thm:mainsimple}, we show that this issue can be resolved {\em by mentioning e-values rather than p-values\/} next to the accept/reject decision. E-values \cite{GrunwaldHK19B,VovkW21,Shafer19,wasserman2020universal,ramdas2023savi} are a recently popularized alternative for p-values that are related to, but far more general than, likelihood ratios.
Importantly, as reviewed in Example~\ref{ex:nplrevar} below, for any NP test with the accept/reject-decision based on a p-value, the exact same test can be implemented by basing the decision on an e-value. Thus there is no a priori reason why one should accompany the decision of a NP test with a p-value rather than an e-value. But, in contrast to the p-value, the e-value has a clear decision-theoretic justification that remains valid if decision tasks are formulated {\em post-hoc}, i.e. after seeing, and in light of, the data. Concretely, after the result of a study has been published, and when new circumstances  prevail, one conceivably might contemplate different actions, with different associated losses, than originally planned. For example, a study about vaccine efficacy ({\sc ve}) in a pandemic may have been set up as a test between null hypothesis $\text{\sc ve} \leq 30\%$ and alternative $\text{\sc ve} \geq 50\%$ \cite{SchureG22}. The original plan was to vaccinate all people above 60 years of age if the null is rejected. But suppose the null  actually gets rejected with a very small p-value $\ll \alpha$, and at the same time the virus' reproduction rate may be much higher than anticipated. Based on {\em both\/} the observed data (summarized by $\pval$) {\em and\/} the changed circumstances, one might now contemplate a new action, vaccinate everyone over 40, with higher losses if the alternative is false and higher pay-offs if it is true. E-values can be used unproblematically for such a post-hoc formulated decision task; p-values cannot. A second example is simply the fact that scientific results are {\em published\/} and remain on record so as to be useful for future deployment:
a company contemplating to produce medication $X$ 
may find a publication about the efficacy of $X$ that is, 
say, 15 years old. Back then, in two independent studies the null (no efficacy) was rejected at the given $\alpha=0.05$, 
but producing $X$ would have been prohibitively expensive so this finding was not acted upon. But recently the company managed a technological breakthrough making production of $X$ much cheaper.
Had $\alpha$ been smaller than $0.01$, they would now decide to take $X$ into production. But now suppose that in both original studies, $\pval < 0.01$ yet $\alpha=0.05$. The upshot of Example~\ref{ex:unknownloss}, Proposition~\ref{prop:typeIsafety} and Theorem~\ref{thm:mainsimple} below is that, if one had observed $S^{-1} < 0.01$ for an e-variable $S$ then acting anyway, despite the changed circumstances is {\em Type-I risk safe}, in the precise sense of \eqref{eq:TypeIrisksafe} below; but doing this based on $\pval < 0.01$ is unsafe in the sense that no clear risk (performance) bounds can be given when engaging in such behavior. 

\paragraph{From Testing to Estimation with Confidence: the e-posterior}
The medication $X$ example is too simplistic: it only deals with rejecting the null of `no efficacy'. In reality one wants to take effect sizes into account as well when making decisions.  Section~\ref{sec:eposterior} shows that e-value methods extend to that setting as well.
Upon observing data from a  statistical model with parameter of interest $\theta$, the question now becomes how to properly interpret the statement ``$\theta \in \textsc{cs}_{\alpha}$'', where $\textsc{cs}_{\alpha}$ is a $(1-\alpha)$-confidence set, usually an interval. The correct, basic interpretation only says that, when repeatedly performing studies, the true parameter will lie in $\textsc{cs}_{\alpha}$ in a fraction of about $1-\alpha$ studies. But practitioners want more, and indeed, $\textsc{cs}$'s are often given an evidential interpretation --- one outputs not one but a system of confidence intervals, one for each of a series of coefficients such as $80\%, 90\%, 95\%, 99\%$, or even a full {\em confidence distribution\/} \cite{Cox58,SchwederH16} and this, it is said ``{\em summarizes what the data tell us about $\theta$, given the model}'' \citep[page 227]{CoxH74} or ``the information about the parameter'' \citep{Lehmann59}. As our second  contribution, we show there is  benefit in replacing standard $\cs$'s by e-$\cs$'s, and confidence distributions by {\em e-posteriors\/} \cite{Grunwald23}: again, these stand on firmer decision-theoretic ground.
\paragraph{BIND Assumption underlying p-values and standard CSs}
While so far we highlighted the problems with post-hoc determined loss functions, below  we show that decisions based on \textsc{p}'s (Section~\ref{sec:gnp}) and \textsc{cs}'s (Section~\ref{sec:eposterior}) may already become unsafe, in the Type-I risk sense of \eqref{eq:TypeIrisksafe}, as soon as the decision task involves a `Type-I' loss function that  can take on more than two values, 
even if this loss function {\em is\/} determined in advance. Essentially, we can only be sure that decisions based on \textsc{p}'s and \textsc{cs}'s are reliable if both (1) the loss function is binary-valued (B) and (2), it is determined in advance, or at least independently (IND) of the observed data. Thus, they really operate under a BIND (binary $+$ independence) assumption. E-values and -posteriors lead to decisions that retain Type-I risk safety if BIND is violated. 

\paragraph{Technical Contribution and Contents}
To obtain frequentist guarantees without BIND we first need to reformulate NP testing in terms of losses and risks rather than errors and error probabilities, an idea going  back to Wald's seminal 1939 paper introducing statistical decision theory \citep{Wald39}. But while Wald lets go off the Type-I/II error distinction as soon as he allows for more than two actions, we stick with Type-I and Type-II risks (replacing Type-I and Type-II error probabilities, respectively)  and show that the e-value is then the natural statistic to base decisions upon, and remains so if the decision task is determined post-hoc. Thus, our {\em GNP\/} ({\em Generalized Neyman-Pearson\/}) Theory follows a path opened up by Wald but apparently not pursued further thereafter. 
In Section~\ref{sec:gnp} we informally present this reformulation, show how $\pval$-based procedures get in trouble if BIND is violated, introduce e-values and explain how, when combined with a {\em maximally compatible decision rule}, they guarantee Type-I risk safety even without BIND. Section~\ref{sec:math} then formalizes the reasoning and presents our main result, Theorem~\ref{thm:mainsimple}. 
Among all Type-I risk safe decision rules, we aim only for those that have {\em admissible\/} Type-II risk behavior; 
we call a rule admissible if there exists no other decision rule that is never worse and sometimes strictly better. Theorem~\ref{thm:mainsimple}, which has the flavour of a  {\em complete class theorem\/} 
\citep{BergerB85,DuanmuR21}  shows that, under mild regularity conditions, 
the set of admissible decision rules are precisely those that are based on some e-variable $S$ via a maximally compatible decision rule. 
Section~\ref{sec:eposterior} extends our findings to confidence intervals and distributions (\cd's).  \cd's can be replaced by e-posteriors, a novel notion treated in much more detail in my recent paper \cite{Grunwald23}, which may be be viewed as a companion to this one, more oriented towards a Bayesian-inclined readership. 
\paragraph{An Important Caveat} Systematic development of e-values 
has only started very recently (in 2019). While a lot of progress has been made, and by now useful ($\approx$ powerful) e-values are available for a number of practically important parametric and nonparametric testing and estimation problems, there is still an enormously wide range of problems for which p-values --- systematically developed since the 1930s --- exist yet e-values have not yet been developed. We briefly review initial success stories and current challenges in Section~\ref{sec:stateoftheart}, 
informing the final Section~\ref{sec:discussion} which indicates the way forward and  re-interprets our findings as establishing a {\em quasi-conditional\/} paradigm.  All longer mathematical derivations and proofs are delegated to the Supporting Information Appendix (SI). 
\section{Generalized Neyman-Pearson Theory }\label{sec:gnp}
\subsection{Losses instead of Errors}
\label{sec:losses}
 In the basic NP setting, we observe data $Y$ taking values in some set $\cY$, with both the null hypothesis $\Hnull$ and the alternative $\Halt$ being represented as collections of distributions for $Y$.
NP \cite{NeymanP33} tell us to fix some $\alpha$ and then adopt the decision rule that, among all decision rules with Type-I error bounded by $\alpha$, minimizes the Type-II error. Following Wald \cite{Wald39}, we
re-interpret this procedure in terms of a nonnegative loss function $\loss{\cdot}{\cdot}$, with $\loss{\kappa}{a}$ denoting the loss made by action $a$ if $\kappa$ is the true state of nature. 
We have 
$\kappa \in \{\snull,\alt\}$ and $\cA = \{ 0,1\}$, 
$\kappa=\snull$ representing that the null is correct, $\kappa=\alt$ that the alternative is correct, $a = 0$ standing for `accept' and $a =1$ for `reject' the null. We invariably assume 
$\loss{\snull}{1} > \loss{\snull}{0} \geq 0, \loss{\alt}{0} > \loss{\alt}{1} \geq 0$. 
`Of course' (as Wald writes) we may want to set $\loss{\snull}{0} = \loss{\alt}{1} = 0$ and we will do this for now, but it is not required for the subsequent developments.
In this formulation, the usual $\alpha$-Type-I error guarantee is replaced by an $\ell$-{\em Type-I risk guarantee\/}. Formally, we fix an $\ell$ in advance of observing the data and we say that decision rule $\delta$ (i.e. a test), defined as a function from $\cY$ to $\cA$, is {\em Type-I risk safe\/} if 
\begin{align}\label{eq:TypeIrisksafe}
& \sup_{P_0 \in \Hnull} {\bf E}_{Y \sim P_0} [\loss{\snull}{\delta(Y)}] \leq \ell, \end{align}
where, for $j=0,1$, $P_j \in \cH(\underline{j})$,  ${\bf E}_{Y \sim P_j} [\loss{\snull}{\delta(Y)}]$ is called the {\em risk of $P_j$}, i.e. the expected loss under $P_j$. 
Following NP again, with again `error probability' replaced by `risk', we now postulate that among all Type-I risk safe decision rules $\delta$,  we ideally want to pick one that has small {\em worst-case Type-II risk\/}
\begin{align}\label{eq:typeIIrisk}
\sup_{P_1 \in \Halt} {\bf E}_{Y \sim P_1} [\loss{\alt}{\delta(Y)}].
\end{align}
\eqref{eq:TypeIrisksafe} expresses that, whatever we decide, we want to make sure that our risk (expected loss) under the null is no larger than $\ell$.
In a standard level-$\alpha$ test, one rejects the null if $\pval(y)$, the p-value corresponding to data $y$, satisfies $\pval(y) \leq \alpha$.
A corresponding decision rule in terms of loss functions is to reject the null whenever the observed $\pval(y)$ satisfies 
\begin{align}\label{eq:pvalpre}
    \pval(y) \cdot \loss{\snull}{1} \leq \ell.
\end{align}
We get exactly the same behavior as for the standard level $\alpha$-test if we set  $\loss{\snull}{1} = \ell/\alpha$. For example, for  $\alpha= 0.05$ we can set $\ell=1$ and then $\loss{\snull}{1} := 20$; then, just like in NP testing \eqref{eq:pvalpre} tells us to pick a $\delta^{\circ}$ which rejects the null if $\pval \leq 0.05$. If $\pval$ is defined so that $\delta^{\circ}$  is UMP (uniformly most powerful), then combined with any loss function $\loss{\alt}{0} >0$, $\delta^{\circ}$ will also minimize the worst-case Type-II risk \eqref{eq:typeIIrisk} among all $\delta$ that satisfy Type-I error probability $\leq \alpha$: up until now we have merely reformulated standard NP theory.
\paragraph{Actions of Varying Intensity} But now suppose we have {\em more than two\/} actions available. For example, consider four alternative actions: accept the null (retain the status quo), take mild action (e.g. vaccinate all people over 60), take more drastic action (vaccinate everyone over 40) and extreme action (vaccinate the whole population). 
We consider this question, too,  in terms of Type-I and Type-II risk and confidence --- thereby taking a different direction than standard decision theory.
For example, our action space could now be $\cA_b = \{0,1,2, 3 \}$  with loss function  $\bloss{b}{\snull}{0} =0, \bloss{b}{\snull}{1}=20 \ell$, $\bloss{b}{\snull}{2} = 100\ell$, $\bloss{b}{\snull}{3} = 500 \ell$ and $\bloss{b}{\alt}{3} < \bloss{b}{\alt}{2} <  \bloss{b}{\alt}{1} <  \bloss{b}{\alt}{0} = \ell$. More generally, as long as  Type-I loss is increasing in $a$ and Type-II loss is decreasing, such an extension of the NP setting makes intuitive sense. 

In terms of p-values, the straightforward extension of \eqref{eq:pvalpre} to this multi-action case would be to play action $a$ where $a$ is the largest value such that \begin{equation}\label{eq:pvalb}
    \pval(y) \cdot \bloss{b}{\snull}{a} \leq \ell.
\end{equation}
But, assuming our p-value is strict so that it has a uniform distribution under the null,  this gives a Type-I risk of 
\begin{multline}\label{eq:losstoolarge}
{\bf E}_{Y \sim P_0}[\bloss{b}{\snull}{\delta(Y)}] = 
\left(\frac{1}{20} - \frac{1}{100}\right) \cdot 20 \ell +  \left(\frac{1}{100}
- \frac{1}{500} \right)
\cdot 100 \ell 
+ \frac{1}{500}  \cdot 500 \ell 
= 2.6 \ell,
\end{multline}
violating the guarantee we aimed to impose and showing that a naive p-value based procedure does not work. The problem gets exacerbated if we allow for more than four actions: in the SI we show that the expected loss of the naive procedure \eqref{eq:pvalb} may go to $\infty$ as we add additional actions with $\bloss{b}{\snull}{a}$ increasing and $\bloss{b}{\alt}{a}$ decreasing in $a$. There we also show that an obvious `fix', namely modifying  \eqref{eq:pvalb} to make sure that for each action $a$, $\bloss{b}{\snull}{a}$ gets multiplied by exactly the probability that action $a$ is taken, does not solve this issue.

\paragraph{Post-Hoc Loss Functions}
Allowing more than two actions is really just a warm-up to a further extension which arguably better models what often happens in, for example, medical practice: 
the post-hoc determination or modification of a decision task, after seeing the data and dependent on the data, such as in the vaccine efficacy example in the introduction.
That is, there is really an underlying class (whose definition may be unknowable) of loss functions $\bloss{b}{\cdot}{\cdot}$ with associated action spaces $\cA_b$,
\commentout{

re are {\em researchers}, who perform scientific studies;  there are {\em policy developers}, who look at the available evidence from the study and at external circumstances and formulate some possible actions, and then finally there is a team of  {\em decision-makers\/}  that have to make a decision, i.e. perform an action (in practice, two or even all three  groups may sometimes coincide). 
For example, a research group may find from a clinical trial that the efficacy of a newly developed vaccine is at least $50\%$, at a significance level $\alpha= 0.01$. 
Then, a country's {\em health council\/} may suggest possible actions (such as vaccination of a  part of the population) afterwards, possibly in complex circumstances that were not anticipated during the trial (e.g. the disease is spreading faster than was hoped for; or public support for lock-downs is declining; or a new variant of the virus is emerging). This is posed to the} 
and the 
decision-maker (DM) is posed a particular decision task $L_b(\cdot,\cdot)$ where $b$, indexing the loss actually used, is really the outcome of a random variable $B=b$, whose distribution may depend on the data in all kinds of ways. The actual $B=b$ that is presented is thus random and only fixed {\em after\/} the study result has become available; i.e. `post-hoc'.  Crucially, the process determining the actual value of $B$ is typically murky; nobody knows exactly what loss function would have been considered in what alternative circumstances; DM only knows the loss function finally arrived at. 

Again, with p-values, we might be tempted to pick the largest action $a$  such that \eqref{eq:pvalb} holds,  
where now $b$ is really the (observed, known) outcome of  random variable $B$ whose definition is itself unknown. 
Now, even if for each $b$, $L_b$ allows for only two actions, so that the problem superficially resembles the standard NP setting, using \eqref{eq:pvalb} can have disastrous consequences in the post-hoc setting, as the following example shows.
\begin{example}\label{ex:unknownloss} Suppose there are three loss functions $L_b$, for $b \in \cB = \{1,2,3\}$, with corresponding actions 
$\cA_b = \{0,b\}$. We set 
$\bloss{1}{\snull}{1} = 20\ell, \bloss{2}{\snull}{2} = 100\ell, \bloss{3}{\snull}{3} =500\ell $, $\bloss{b}{\snull}{0} = 0, \bloss{b}{\alt}{0} := \ell$ for all $b \in \cB$, and $\bloss{b}{\alt}{b}$ strictly decreasing in $b$.
This is like the previous example, but rather than always being able to choose one among four actions, the very set of choices  that is presented to  DM via setting $B=b$ might depend on the data $Y$ or on external situations.
One cannot rule out that this is done in an unfavourable manner --- if the data suggest strong evidence then the policy developers (e.g. a pandemic outbreak management team) might only suggest actions with drastic consequences. Suppose, for example, that if $\pval> 0.02$, the DMs are presented loss $L_1$; if $0.001 < \pval \leq 0.02$ they are presented loss $L_2$; and if $\pval \leq 0.001$ they are presented loss $L_3$. Using \eqref{eq:pvalb}, we then get (assuming again uniform $\pval$) a Type-I risk of 
\begin{multline}
{\bf E}_{Y \sim P_0}[\loss{\snull}{\delta(y)}] = 
(0.05- 0.02) \cdot 20 \ell +  (0.02-0.001) \cdot 100 \ell + 0.001 \cdot 500 \ell=
3 \cdot \ell. \nonumber
\end{multline}
As in \eqref{eq:losstoolarge} the resulting decision rule \eqref{eq:pvalb} is not Type-I risk safe, and again, the Type-I risk can even go to infinity with the number of potential actions. 
\end{example}
\subsection{E-Values to the Rescue}
\label{sec:erescue} 
Reporting evidence as e-values (as defined by \eqref{eq:basic}) rather than p-values solves both the multiple action and post-hoc-loss issue identified above. 
An e-{\em value\/} is the value of a special type of statistic called an e-{\em variable}. An e-variable is any {\em nonnegative\/}
random variable $S=S(Y)$ that can be written as a function of the observed $Y$ and that satisfies the inequality:
\begin{equation}\label{eq:basic}
\text{for all $P \in \Hnull$:}\ \ {\bf E}_{P}[S] \leq 1.
\end{equation}
The e-variable's simplest application is in defining tests: the $S$-based hypothesis test at level $\alpha$ is defined to reject the null iff $S \geq 1/\alpha$. Since for any e-variable $S$, all $P \in \Hnull$, by Markov's inequality, $P(S \geq 1/\alpha) \leq \alpha)$, with such a test we get a Type-I error guarantee of $\alpha$, with the advantage that (as shown by \cite{GrunwaldHK19B,VovkW21}) unlike with p-values, the Type-I error guarantee remains valid under {\em optional continuation}, i.e. deciding based on a study result whether new studies should be undertaken and if so,  multiplying the corresponding e-values.  The term `e-variable' was coined in 2019 \cite{GrunwaldHK19B,VovkW21} but their history is older, as described by \cite{GrunwaldHK19B,ramdas2023savi}. 

We may now simply pick any e-variable $S$ we like and replace decision rule \eqref{eq:pvalb} by the following {\em 
maximally compatible\/} alternative rule:
upon observing data $Y=y$ and loss function indexed by $B=b$ with accompanying maximum imposed risk bound $\ell$, select the {\em largest\/} $a$ for which
\begin{align}\label{eq:evaldr}
 S^{-1}(y) \cdot \bloss{b}{\snull}{a} \leq \ell, 
 \text{\ i.e. \ }
    \bloss{b}{\snull}{a} \leq S(y) \cdot \ell,
    \end{align}
where
we adopt the (in our setting harmless) convention that, for  $u= 0$ and $v \geq 0$,  $u^{-1} v := 0$ if $v=0$ and $u^{-1} v = \infty$ if $v> 0$ (in Section~\ref{sec:discussion} we discuss where $\ell$ comes from).
Theorem~\ref{thm:mainsimple} below gives conditions under which \eqref{eq:evaldr} has a 
unique solution. 
For the original  NP setting of two actions, \eqref{eq:evaldr} is simply the p-value based rule \eqref{eq:pvalb} with $\pval$ replaced by $1/S$, illustrating that 
{\em large\/}
e-values correspond to evidence against the null. But in contrast to the p-value based rule, 
with the e-based rule, no matter what e-variable $S$ we take (as long as it is itself chosen before data are observed), no matter how many actions $\cA$ contains, no matter the process determining the loss $B$, we have the Type-I risk guarantee \eqref{eq:TypeIrisksafe} (Theorem~\ref{thm:mainsimple} below): replacing $\pval$ by $1/S$
resolves the BIND problem.
Of course, this raises the question whether p-values cannot be used safely for Type-I risk after all, in a manner different from \eqref{eq:pvalb}. The only such method we know of 
is to first convert a p-value into an e-value and then use 
\eqref{eq:evaldr} after all. As discussed in the SI, the e-values resulting from such a conversion are usually suboptimal, so we prefer to design and use e-values directly.  


\begin{example}{\pnasfont [The NP and LR E-Variables]}\label{ex:nplrevar}
As with p-values, many different e-variables can be defined for the same $\Hnull$. As discussed by \cite{Shafer19}, an extreme choice is to start with a fixed level $\alpha$ and p-value $\pval$ and to set $\Snp{\alpha} := (1/\alpha)$ if $\pval \leq \alpha$ and $\Snp{\alpha}=0$ otherwise. Clearly 
${\bf E}_{Y \sim P_0}[ \Snp{\alpha}] \leq \alpha (1/\alpha) = 1$ so $\Snp{\alpha}$ is an e-variable. 
In the case of a classical, 2-action NP problem as defined underneath \eqref{eq:pvalpre}, the test \eqref{eq:evaldr} based on e-variable $S= \Snp{\alpha}$ will lead to $a=1$ (reject the null) exactly iff the classical NP test based on $\pval$ does. This shows that any $\pval$-based NP test can also be arrived at using \eqref{eq:evaldr} with a special e-value: nothing is lost by replacing p-values with e-values. Still, in case there are more than 2 actions and/or post-hoc decisions, while  preserving the $\ell$-Type-I risk guarantee, decisions based on  $\Snp{\alpha}$ may not be very wise in the Type-II risk sense. For example, with the loss function used in \eqref{eq:losstoolarge} and $\alpha = 0.05$, we get that even for very small underlying $\pval$ (i.e. extreme data), we will still choose action $1$ whereas it seems more reasonable to select more extreme actions, minimizing Type-II loss, as the evidence against the null gets stronger. 
In case $\Hnull = \{P_0\}$ and $\Halt =\{P_1\}$ are simple, this can be achieved by taking $S$ to be a {\em likelihood ratio\/}: assuming $P_j$ has density $p_j$,
\begin{equation}
    \label{eq:lr}
    \Slr := \frac{p_1(Y)}{p_0(Y)} 
\end{equation}
which is immediately seen to be an e-variable: 
\begin{equation}\label{eq:eproof}
{\bf E}_{P_0}[\Slr]= \int p_0(y) \frac{p_1(y)}{p_0(y)} d\mu =\int p_1(y) d\mu = 1,
\end{equation}
i.e. to satisfy \eqref{eq:basic}. We can compare $\Snp{\alpha}$ and $\Slr$ if $\pval$ underlying $\Snp{\alpha}$ is itself a monotonic function of the likelihood ratio $\Slr$, as it will be for the standard optimal power NP test. In the decision task above \eqref{eq:pvalb}, when used in \eqref{eq:evaldr},  $\Snp{\alpha}$ can, for each $\alpha$, select at most 2 actions whereas    $\Slr$  leads to selection of action $0,1,2$ or $3$ depending on the amount of evidence, at the price of imposing a larger threshold before any particular action is selected compared to the $S^{\alpha}$ that is optimal for that action (e.g. $S^{0.05}$ is optimal for action $1$ in this sense). We will see more sophisticated  e-variables in Example~\ref{ex:prebexample} and refer to numerous further examples of useful e-variables in Section~\ref{sec:stateoftheart}.
\end{example}
\section{Mathematical Formalization and Results }\label{sec:math}
\subsection{Type-I Risk Safety and Compatibility}\label{sec:TypeI} 
Let $\Hnull$, the null hypothesis, be a set of probability distributions for random $Y$ taking values in a {\em outcome  space $\cY$}.
\begin{definition}\label{def:gnptest}
A {\em GNP (Generalized Neyman-Pearson) testing problem\/}  relative to $\cH(\snull)$ is  a tuple $(\cB,\{(\cA_b, L_b(\snull,\cdot):  \cA_b \rightarrow \reals^+_0): b \in  \cB \})$ where  for all $b \in \cB$,  
we call $\bloss{b}{\snull}{\cdot}$ the {\em Type-I loss indexed by $b$ with action space $\cA_b$}.
\end{definition}
In Section~\ref{sec:eposterior} we extend the definition to uncertainty quantification beyond testing.
Relative to any given GNP testing problem,  we further define a {\em decision rule\/} to be any collection of functions $\{\delta_b: b \in \cB\}$, where $\delta_b(y)$  denotes the $a \in \cA_b$ picked when loss function indexed by $B=b$ (i.e. $L_b$) is presented and $Y=y$ is observed.
Let $\delta$ be any decision rule and let $S=S(Y)$ be any e-variable. We call $\delta$ {\em compatible with  $S$} if 
\begin{equation}\label{eq:compatible}
\bloss{b}{\snull}{\delta_b(y)} \leq S(y) \text{\ for all $y \in \cY$, $b \in \cB$}.
\end{equation}
We now prepare the definition of Type-I risk safety for GNP decision problems. First, we note that in general, the threshold $\ell$ a DM would like to impose on the risk via \eqref{eq:evaldr} when confronted with loss function $L_b$ may be an arbitrary positive real. 
However, using this maximal rule \eqref{eq:evaldr}, for every observed $Y=y$ and $B=b$, the exact same decision will be taken if we normalize all losses, using $L'_b$ with $L'_b(\snull,a) = L_b(\snull,a)/\ell$ instead of $L_b$ and $\ell'=1$ instead of $\ell$. Hence, without loss of generality, from now on we simplify the treatment by taking $\ell =1$ (in the SI  we discuss in more detail why this is not harmful).
With this in mind, consider a concrete setting in which the actual loss function $L_B$ with index $B$ presented to  DM is determined in a data-dependent manner (perhaps by some policy makers, perhaps completely implicitly). Since we do not know the definition of $B$, i.e. how the choice is made, we want to ensure that the analogue of Eq.~\ref{eq:TypeIrisksafe} holds, in the worst case, over all possible choices. Thus, as a first attempt, we may extend Eq.~\ref{eq:TypeIrisksafe},
by defining $\delta$ to be Type-I risk safe if
\begin{equation}\label{eq:TypeIRiskSafeSimplified}
\sup_{P_0 \in \cH(\snull)} {\bf E}_{P_0} \left[ \sup_{b \in \cB} {L_{b}(\snull,\delta_b(Y))}\right] \leq 1.\end{equation}
As discussed in the SI, the expectation might be undefined for pathological choices of the set $\cB$ and the 
functions $\{L_b: b \in \cB\}$ --- after all, we have not restricted the choice of $\cB$ and $L_b$ at all. We can simply avoid this issue by slightly modifying the definition: 
we define $\delta$ to be {\em Type-I risk safe\/} if there exists a function $U: \cY \rightarrow \reals^+_0$ such that for all $P_0 \in \cH(\snull)$, ${\bf E}_{P_0}[U(Y)]$ is well-defined, and  
for all $y \in \cY$,
\begin{align}\label{eq:TypeIRiskSafeSimplifiedB}
\hspace*{-0.22 cm} \sup_{b \in \cB} {L_{b}(\snull,\delta_b(y))}
\leq U(y) \ \text{;} 
\sup_{P_0 \in \cH(\snull)} {\bf E}_{P_0} \left[ U(Y) \right] \leq 1.\end{align}
\paragraph{E-Variable Compatibility $\Leftrightarrow$ Type-I Risk Safety}
In NP Theory, Type-I error guarantees come first --- we look for an optimal decision rule among all rules that have the desired Type-I error guarantee. Analogously, here we first restrict our search for `good' decision rules to those that are Type-I risk safe for the given decision problem. How to find these? Realizing that the second equation in \eqref{eq:TypeIRiskSafeSimplifiedB} expresses that $U$ is an e-variable, and the first equation says that $\delta$ is compatible with this e-variable, we see that the Type-I risk safe decision rules are exactly those that are compatible with an e-variable, thereby explaining the importance of e-variables to generalized NP testing. Formally, we have just proved the following trivial consequence of our definitions: 
\begin{proposition}
\label{prop:typeIsafety}
Fix an arbitrary  GNP testing problem.  For every $\delta$ defined relative to this problem: \begin{enumerate}
    \item  For every e-variable $S$ for $\Hnull$:  if  $\delta$ is compatible with $S$, then $\delta$ is Type-I risk safe. 
\item Suppose that $\delta$  is 
Type-I risk safe. Let $S=U$ be as in \eqref{eq:TypeIRiskSafeSimplifiedB} (in standard cases we can simply take $S(y) = \sup_{b \in \cB}  \bloss{b}{\snull}{\delta_b(y)}$). Then $S$ is an e-variable for $\Hnull$, and $\delta$ is compatible with  $S$. 
\end{enumerate}  
\end{proposition}
\subsection{Admissibility}\label{sec:TypeII}
We now turn to Type-II losses. The reader may have wondered why the specification of Type-II loss functions $L_b(\alt,a): \cA_b \rightarrow \reals$ as in \eqref{eq:typeIIrisk} was not made part of Definition~\ref{def:gnptest}. The following crucial observation implies that this is superfluous, thereby greatly satisfying the treatment: suppose there were two actions $a,a' \in \cA_b$ such that $L_b(\snull,a' ) > L_b(\snull,a)$ and   $L_b(\alt,a' ) > L_b(\alt,a)$. Then any rational DM would always prefer $a$ over $a'$, and hence never want to play $a'$. We can thus take $a'$ out of the set $\cA_b$ without affecting the set of decisions that a DM might ever want to consider. Assuming that $\cA_b$ has been pre-processed like this, we automatically obtain that the larger the Type-I loss of an action, the smaller the Type-II loss, allowing us to refrain from specifying $L_b(\alt,\cdot)$: we may thus call a decision rule $\delta^{\circ}$ {\em Type-II strictly better\/} than $\delta$ if for all $b \in \cB$, all $P \in \cH(\snull)$, 
we have 
\begin{equation}\label{eq:TypeIIbettera}
P(L_b(\snull,\delta^{\circ}_b(Y)) <  L_b(\snull,\delta_b(Y))) = 0
\end{equation}
whereas there exist $b\in \cB, P \in \cH(\snull)$ such that
\begin{equation}\label{eq:TypeIIbetterb}
P(L_b(\snull,\delta^{\circ}_b(Y)) >  L_b(\snull,\delta_b(Y))) > 0.
\end{equation}
If $\cY$ is uncountable,  $L_b$, $\delta^{\circ}$ and $\delta$ could again be picked in highly pathological ways, such that the probabilities above are undefined. This is fully resolved by the generalization of \eqref{eq:TypeIIbettera} and \eqref{eq:TypeIIbetterb} given in the SI.

Clearly, if both $\delta^{\circ}$ and $\delta$ are Type-I risk safe and $\delta^{\circ}$ is Type-II strictly better than $\delta$, we would always prefer playing $\delta^{\circ}$ over $\delta$. We may say that $\delta$ is {\em inadmissible}.
Formally, for any decision rule $\delta$ we say that it is {\em admissible\/}  if it is Type-I risk safe and no other Type-I risk safe decision rule is Type-II strictly better. 
\paragraph{Main Result}
This admissibility notion is reminiscent of standard admissibility notions in classical statistical decision theory, and the theorem below is in the spirit of a {\em complete class theorem\/} \citep{BergerB85,DuanmuR21} expressing that in searching for reasonable (i.e., admissible) decision rules in GNP problems we may restrict ourselves to those based on e-variables via  {\em maximally compatible decision rules}. Formally, 
we call a decision rule $\delta$  {\em maximally compatible\/} with e-variable $S$ relative to a given GNP testing problem, if it is compatible with $S$ and there exists no decision rule $\delta^{\circ}$ such that $\delta^{\circ}$ is also compatible with $S$ yet $\delta^{\circ}$ is Type-II strictly better than $\delta$. 
We will relate this to the earlier informal definition of maximum compatibility (\eqref{eq:evaldr}) further below. 

To state the theorem, we need one more concept: we call a GNP testing problem  
{\em rich\/} relative to e-variable $S=S(Y)$ if for every $s$ in the co-domain of $S$, there exist $b \in \cB$ and $a \in \cA_b$ such that $L_b(\snull,a)=s$. An example of a simple GNP testing problem that is rich relative to any e-variable at all is obtained whenever $\cB= \{\textsc{sq}\} \cup \cB'$, for arbitrary $\cB'$, where  $\cA_{\textsc{sq}} = \reals^+_0$ and $L_{\textsc{sq}}(\snull,a) = a^2$ (the squared error loss --- richness follows since it can take on any value in $\reals^+_0$). An example of a GNP testing problem that is rich relative to e-variable $S^{\textsc{np}(\alpha)}$ of Example~\ref{ex:nplrevar}   is given by $\cB= \{\textsc{np}\} \cup \cB'$, for arbitrary $\cB'$, where  $\cA_{\textsc{np}} = \{0,1\}$ and $L_{\np}(\snull,0)= 0, L_{\np}(\snull,1) = 1/\alpha$ (if $\cB' = \emptyset$, this is the classical NP setting of Section~\ref{sec:gnp} again): choose $B= 
\np$, $a=0$ if $S^{\np(\alpha)} =0$, and choose $B=\np$, $a=1$ if $S^{\np(\alpha)} = 1/\alpha$. 
\begin{theorem}\label{thm:mainsimple}
Consider a GNP testing problem.  
Then:
\begin{enumerate}  
\item
If $\delta$ is an admissible
decision rule, then there exists an e-variable $S$ such that $\delta$ is a maximally compatible decision rule for $S$.
\item As a partial converse, suppose that $\delta$ is a maximally compatible decision rule for some e-variable $S$. If 
(a) all $P \in \cH(\snull)$ are mutually absolutely continuous (see below) and (b) $S$ is sharp relative to the given testing problem,  i.e. ${\bf E}_{P_0}[S] = 1$ for some $P_0 \in \Hnull$, and (c) 
the GNP testing problem is rich relative to $S$, then $\delta$ is admissible. 
\end{enumerate}
\end{theorem}
Part 1 shows that we can restrict our search for admissible decision rules to the ones that are maximally compatible for some e-variable $S$. Part 2  is in essence a converse, showing that, under some  regularity conditions, maximally compatible decision rules must be admissible. 
All three conditions required are weak: (a) Two distributions $P, P'$ are mutually absolutely continuous if `they agree on what is practically impossible', i.e. for each  event $\cE$, we have $P(\cE) = 0$ iff $P'(\cE)=0$. Most standard parametric families are absolutely continuous or can be made such by excluding the boundary of the parameter space. (b) Sharpness of $S$ expresses that $S$ cannot be uniformly improved --- a mild requirement satisfied by all e-variables considered in this paper (and also in most other papers on e-variables \cite{ramdas2023savi,GrunwaldHK19B}). 
(c) Richness relative to the $S$ considered holds in all examples encountered in this paper (see Example~\ref{ex:admissible} below for further illustration). More importantly perhaps, for any sharp e-variable $S$ which we might want to base our decisions on, we can trivially {\em enlarge\/} any given
GNP testing problem by adding one particular loss function so that 
the extended GNP decision problem will automatically be rich relative to $S$, and Part 2 of Theorem~\ref{thm:mainsimple} can then  be applied. 
In the SI we explain how this enlargement works and why it is a reasonable operation.

The theorem thus expressing that maximally compatible $\delta$ tend to coincide with admissible $\delta$, we would still like to be assured that such maximally compatible $\delta$ exist in wide generality. 
To briefly illustrate that this is the case, at the same time connecting the formal notion to the earlier informal definition based on \eqref{eq:evaldr}, let us consider what we will call {\em simple\/} GNP testing problems. All GNP testing problems encountered in this and the previous section are simple. They are defined as those GNP testing problems for which, (i) for all $b \in \cB$, $\cA_b$ is a finite union of closed intervals in $\reals^+_0 \cup \{\infty \}$ (in particular this includes the case that $\cA_b$ is finite); (ii) the Type-I loss $L_b(\snull,a)$ is monotonically and, on each interval, continuously increasing in $a$, and (iii) all $P \in \cP$ are mutually absolutely continuous.  The following is  easily checked: for arbitrary e-variable $S$, such simple GNP decision problems must have a maximally compatible $\delta^*$ relative to $S$ that generalizes  \eqref{eq:evaldr}, with our simplification $\ell=1$: $\delta^*$ is the rule which selects,
when presented $Y=y,B=b$, 
\begin{align}
    \label{eq:maxcompatible}
  \delta^*_b(y) := \text{largest $a \in \cA_b$ with
  $ \bloss{b}{\snull}{a} \leq  S(y)$}.
\end{align}
Moreover, this maximally compatible $\delta^*$ is essentially unique, i.e. if $\delta^*, \delta$ are both maximally compatible, then for all $P \in \cH(\snull)$, we have $P(\delta^* \neq \delta') = 0$.
\begin{example}\label{ex:admissible}
Consider a simple vs. simple testing problem with  $\Hnull = \{P_0\}, \Halt = \{P_1\}$. Let $\pval(Y)$ be a strict p-value, i.e. $P_0(\pval \leq\alpha) = \alpha$ for $\alpha\in [0,1]$, that is monotonically and continuously decreasing in the likelihood ratio $\Slr(Y)$; use of  such a p-value is standard in NP testing with continuous-valued outcome spaces.
Consider the following variation of Example~\ref{ex:unknownloss}: $\cB \subset \reals^+_0$ with 
for $b \in \cB$, $\cA_{b} = \{0,1\}$ and $L_{b}(\snull,0) = 0, L_{b}(\snull,1)= b$.
Take arbitrary but fixed $0< \alpha < 1$. Then the  maximally compatible decision rule $\delta^*$ as in \eqref{eq:maxcompatible} relative to e-variable $\Snp{\alpha}$ is sharp. When presented with loss function $L_b$, this $\delta^*$ always plays $0$ if $b >  1/\alpha$. If $b \leq 1/\alpha$, it plays $1$ if
$b \leq \Snp{\alpha}$ (i.e. if $\Snp{\alpha} = 1/\alpha$, i.e. if $\pval \leq \alpha$) and $0$ otherwise (i.e. if $\Snp{\alpha} = 0$, i.e. if $\pval > \alpha$). By Part 2 of  Theorem~\ref{thm:mainsimple}, this $\delta^*$ is admissible if $\cB$ contains $b= 1/\alpha$, which ensures richness relative to $S^{\np(\alpha)}$. 

In contrast, consider the  $\delta^*$ as in \eqref{eq:maxcompatible} based on the likelihood ratio e-variable $\Slr$, which is also sharp. When  presented $L_b$, this decision rule plays $1$ if $b \leq  \Slr$ and $0$ otherwise. If we set $\cB= \reals^+_0$, we have richness relative to $\Slr$ so by Theorem~\ref{thm:mainsimple}, this $\delta^*$ is admissible as well. In this case though, admissibility of $\delta^*$ may fail if we take $\cB$ a strict subset of $\reals^+_0$. 
\end{example}

\renewcommand{\lr}{\ensuremath{\textsc{lr}}}

\section{Robust Confidence via the  E-Posterior}
\label{sec:eposterior}

Now let us consider a statistical model $\cP$ partitioned according to a parameter of interest $\theta\in \Theta$, with $\phi: \cP \rightarrow \Theta$ indicating the parameter corresponding to each $P$; for example, $\theta = \phi(P)$ might be the mean of $P$, or, if $\cP = \{P_{\theta}: \theta \in \Theta\}$ is a parametric model, $\phi$ might simply denote the parameterization function, $\phi(P_{\theta}) = \theta$.
Any collection of  p-values $\{ \pval_{\theta}: \theta \in \Theta \}$, with $\pval_{\theta}$ a p-value for the null $\cH(\theta) := \{P \in \cP: \phi(P) = \theta \}$ can be used to build a valid $(1-\alpha)$ confidence set, by setting $\cs_{\alpha}(Y) = \{\theta: \pval_{\theta}(Y) > \alpha \}$ to be the set of $\theta$'s that would not have been rejected at the given level $\alpha$. For simplicity, we restrict attention to scalar $\Theta \subseteq \reals$; then the $\cs_{\alpha}$ will usually be intervals, and indeed this p-value based construction is a standard way to construct such intervals. 
Analogously \cite{Xu22,ramdas2023savi}, any {\em e-collection}, i.e. a collection of e-variables $\{S_{\theta}:\theta \in \Theta \}$ such that $S_{\theta}$ is an e-variable for the `null' 
$\cH(\theta)$ (by this we mean that $S_{\theta}$ must satisfy \eqref{eq:basic}, i.e. ${\bf E}_{P}[S_{\theta}] \leq 1$, for all $P \in \cH(\theta)$)
can be used to build an equally valid, usually larger, {\em e-based\/} $(1-\alpha)$-confidence set (again, for scalar $\theta$ this usually becomes an interval), one for each $\alpha$, by setting 
$\cs^*_{\alpha}(Y) = \{\theta: S_{\theta}(Y) < 1/\alpha \}$ as the set of $\theta$'s that would not have been rejected at level $\alpha$ with an e-value based test. Below we first give a simple example. We then, in Section~\ref{sec:eposterior}.\ref{sec:retrace} retrace the steps of Section~\ref{sec:gnp} and Section~\ref{sec:math}, re-interpreting confidence sets in terms of actions with associated losses and risks. 
Section~\ref{sec:eposterior}.\ref{sec:obayes} and \ref{sec:eposterior}.\ref{sec:eposteriorsucceeds} 
show  that, once again, if  losses are determined post-hoc (BIND is violated), then standard confidence intervals loose their validity whereas e-based confidence intervals remain Type-I risk safe. Relatedly, without BIND, decisions based on {\em confidence distributions\/} can be unsafe, but those based on the e-posterior --- a means of summarizing e-\cs's for all $\alpha$'s at once --- remain Type-I risk safe. 
\begin{example}\label{ex:prebexample}
Consider the normal location family: data are $Y = X^n$ where,
under $P_{\theta}$,
$X^n =(X_1, \ldots, X_n)$   are i.i.d. $\sim N(\theta,1)$. 
We consider e-based confidence intervals based on various e-collections. 
\cite{Grunwald23} gives various suitable collections, but for simplicity we here stick to a single, simple choice, taken from Example 8 of \cite{Grunwald23}, that, like the standard CI, is 
symmetric around the MLE $\hat\theta(X^n) = n^{-1} \sum X_i$. 
Fix {\em anticipated\/} sample size $n^*$ and 
confidence level $0 < \alpha^*< 1$.
For each $\mean$ we define $\meanmin < \mean$ and $\meanplus > \mean$ to satisfy
\begin{equation} \label{eq:nstarfun}
  \frac{1}{2} 
n^* (\mean - \meanplus)^2 = \frac{1}{2}
n^* (\mean - \meanmin)^2 =
\log \frac{2}{\alpha^*}.
\end{equation}
Now define e-variables $S^-_{\theta}(y) = p_{\meanmin}(y)/p_{\mean}(y)$, $S^+_{\theta}(y) = p_{\meanplus}(y)/p_{\mean}(y)$ and $S_{\mean}(y) = (1/2) ( S^-_{\theta}(y)+ S^+_{\theta}(y))$. These choices can be motivated based on the fact that $S^-_{\theta}$ and $S^+_{\theta}$
are also uniformly most powerful Bayes factors \cite{johnson2013uniformlyB,Grunwald23} and hence reasonable e-variables for 1-sided $\cs$s. We continue with $S_{\theta}$ for two-sided $\cs$s.
As is proved analogously to \eqref{eq:eproof}, $S_{\theta}$ remains an e-variable even if neither the actual sample size $n$ nor the to-be-used significance level $0 < \alpha < 1$  are equal to the hoped-for  $n^*$ and $\alpha^*$; more on this below. In the SI we show that 
a sufficient condition for $S_{\mean}(Y) \geq \alpha^{-1}$, i.e. for $\theta \not \in \cs^*_{\alpha}(Y)$ is that 
\begin{align}
   %
\left| \mean - \hat\mean\right| \geq 
      \sqrt{\frac{2}{n} \cdot  \log \frac{2}{\alpha}}  \cdot g(c), 
 \ \text{with} \nonumber 
 \end{align}\begin{align} \label{eq:normalrejectb}
c = \frac{n^*}{n^{\ }} \cdot \frac{\log (2/\alpha^{\ })}{ \log (2/\alpha^{*})}\ , \ 
g(c) =  \frac{1}{2} \left({c^{1/2} + c^{-1/2}}\right).
\end{align}
As explained in the SI, for fixed $\alpha$,  \eqref{eq:normalrejectb} is tight for all but the smallest $n$. 
Thus, the e-based confidence interval $\cs^*_{\alpha}(Y)$ has width $|\theta - \hat\theta | \asymp 1/\sqrt{n}$, of the same order as the region for the standard Neyman-Pearson test, with a factor $g(c)$ depending on how well aligned $n, n^*, \alpha$ and $\alpha^*$ are: $g(c)$ and hence  the width is minimized, if $c=1$ (and then $g(c)=1$) which is the case if $n=n^*$ and $\alpha=\alpha^*$.  At $\alpha = .05$, we get in this optimal case that  $S_{\theta}(Y) \geq \alpha^{-1}$ if $|\theta - \hat\theta| \geq \sqrt{ (2 \log 40)/n} \approx 2.72/\sqrt{n}$, 
making the e-based  CI wider than the standard CI by a constant factor of $\approx 2.72/1.96 \approx 1.4$ (see Figure~\ref{fig:eposterior}).  
\end{example}
\paragraph{E-Processes} Why do we not simply set $n^*$ in the definition of  $S_{\mean}$ actual to the actual $n$? The reason is that we allow the actual $n$ to be unknown in advance, and even to be  random (i.e. a stopping time with unknown definition): it is easily seen that, if $Y = X^{\tau}$ with $X^{\tau} = (X_1, \ldots, X_{\tau})$ for a stopping time $\tau$ (whose definition may be unknown to DM), then $S_{\mean}(X^{\tau})$ is still an e-variable. Formally, $S_{\mean}(X^1), S_{\mean}(X^2), \ldots$ constitutes an {\em e-process\/} in the sense of \cite{ramdas2020admissibleB,ramdas2023savi}. Thus, we can use $S_{\mean}$ without knowing the definition of  $\tau$, and in particular, $\tau$ may be unequal to $n^*$ --- such as an extension of e-variables to be used with arbitrary stopping times is often, but not always possible \cite{GrunwaldHK19B}; whenever it is, it provides an additional bonus over use of standard p-variables in testing, which require the stopping time to be set in advance. 
As stated, assuming that we base the $S_{\theta}$ on the correct $n$ and $\alpha$, this e-based confidence interval is about $1.4$ times as wide as the standard one; the inevitable (yet, I feel, worthwhile!) price to pay for the added flexibility and robustness: in contrast to the standard one, we can use the e-based interval for unknown $n$ (or $\tau$) as well, and we can also use it to get valid confidence intervals for $B$ if BIND is violated, 
as we proceed to show. 
\subsection{Reformulating Coverage in terms of Type-I Risk}
\label{sec:retrace}
We now generalize the definition of GNP testing problem so that (besides much else) it also allows for estimation with confidence intervals. \begin{definition}\label{def:gnpgeneral}
Fix a set of distributions $\cP$ for $Y$, a set $\Theta$ and  a  function $\phi: \cP \rightarrow \Theta$  mapping $P \in \cP$ to property $\phi(P) \in \Theta$ as above.  A {\em GNP (Generalized Neyman-Pearson) decision problem\/} relative to $\cP$, $\Theta$ and $\phi$  is a tuple $(\cB,\{(\cA_b,L_b:  \Theta \times \cA_b \rightarrow \reals^+_0): b \in \cB\})$.
\end{definition}
A GNP decision problem is really a set of GNP testing problems, one for each $\theta \in \Theta$: we recover Definition~\ref{def:gnptest} by taking a singleton  $\Theta = \{\snull\}, \cP = \Hnull$ and $\phi(P) = \snull$  for all $P \in \Hnull$.
For general $\theta \in \Theta$, the {\em $\theta$-testing problem corresponding to the GNP decision problem\/} is the testing problem  $(\cB,\{(\cA_b,L_b(\theta, \cdot): \cA_b \rightarrow \reals^+_0): b \in \cB\})$ with null hypothesis $\cH(\theta)= \{P: \phi(P) = \theta \}$ and with $L_{b}(\theta,\cdot)$ in the role of $L_{b}(\snull,\cdot)$. All definitions for GNP testing problems are now easily extended to GNP decision problems by requiring them to hold for the corresponding $\theta$-testing problem, for all $\theta \in \Theta$.  
In particular, we say that decision rule $\delta$ is compatible with e-collection  $ \{S_{\theta} : \theta \in \Theta\}$ if  we have for all $y \in \cY$, $b \in \cB$ that
\begin{align}\label{eq:compatibleB}
\forall \theta \in \Theta:    L_b(\theta,\delta_b(y)) \leq S_{\theta}(y).
\end{align}
The definition of Type-I risk safety is extended analogously from \eqref{eq:TypeIRiskSafeSimplifiedB}: 
$\delta$ is Type-I risk safe iff there exists an e-collection $\ecol= \{S_{\theta}: \theta \in \Theta \}$ such that $\delta$ is compatible with $\ecol$. 
If the expectation below is well-defined (which it will be in the confidence interval setting), Type-I risk safety is then clearly equivalent to the corresponding generalization of \eqref{eq:TypeIRiskSafeSimplified}:
\begin{equation}\label{eq:TypeIrisksafeB}
\sup_{\theta \in \Theta} \sup_{P \in \cH(\theta)}  {\bf E}_{P} \left[ {\sup_{b \in \cB} L_{b}(\theta,\delta_b(Y))}\right]  \leq 1.  
\end{equation}
Admissibility is extended analogously: 
we call a decision rule $\delta^{\circ}$ {\em Type-II strictly better\/} than  $\delta$ if 
for all $\theta \in \Theta$, the corresponding $\theta$-testing problem 
satisfies \eqref{eq:TypeIIbettera} with $\snull$ replaced by $\theta$, whereas there exist $\theta \in \Theta$, $b\in \cB, P \in \cH(\snull)$ such that the corresponding $\theta$-testing problem 
satisfies \eqref{eq:TypeIIbetterb} with $\snull$ replaced by $\theta$. 
The definition of admissibility and maximum compatibility are now based on this extended notion of Type-II strictly-betterness and otherwise unchanged; we further extend the notions of  sharpness and richness to this generalized setting and provide a generalization of Theorem~\ref{thm:mainsimple} to full  GNP decision problems in the SI Appendix.  

\paragraph{Confidence Intervals as Actions}
We now instantiate the above to estimation of confidence intervals.
Given a probability model $\cP$ and parameter of interest $\theta\in \Theta \subset \reals$ with $\theta = \phi(P)$ as above, take the GNP decision problem with  this $\Theta$ and $\phi$, and with  $\cB = [1,\infty)$, $\cA_b= \{[\theta_L,\theta_R]: \theta_L, \theta_R \in \Theta,\theta_L \leq \theta_R\}$,
\begin{align}  \label{eq:typeIIconfidence}
   L_b(\theta,[\theta_L,\theta_R]) = b \cdot  {\bf 1}_{\theta \not \in [\theta_L,\theta_R]}.
\end{align}
Thus, we incur a Type-I loss, if the sampling distribution $\theta$ is not in the interval $[\theta_L,\theta_R]$ we specified, and $b$ determines how bad such a mistake is --- this may again be data-dependent: we assume once again that we are presented $B=b$ via a random and potentially unknowable process,
and we want to obtain the Type-I risk  guarantee \eqref{eq:TypeIrisksafeB}, which instantiates to
\begin{equation}
    \label{eq:above}
\sup_{\theta \in \Theta} \sup_{P\in \cH(\theta)} 
 {\bf E}_{P}[\sup_{b \in \cB} b \cdot  {\bf 1}_{\theta \not \in \delta_b(Y)}] \leq 1,
\end{equation}
where 
$\delta_b(Y) =[\theta_L(Y,b),\theta_R(Y,b)]$. Among all decision rules (i.e. confidence intervals) $\delta$ satisfying \eqref{eq:above}, we want to find the narrowest ones. Our definition of Type-II strictly better above `automatically' accounts for this: the extended definition of Type-II betterness implies that $[\theta_L,\theta_R]$ is Type-II strictly better than $[\theta'_L,\theta'_R]$ iff $[\theta_L,\theta_R]$ is a proper subset of $[\theta'_L,\theta'_R]$. 

If we may assume that BIND  holds we can take the supremum over $\cB$ in \eqref{eq:above}  out of the expectation, i.e. 
$b {\bf E}_{P}[{\bf 1}_{\theta \not \in \delta_b(Y)}] \leq 1$ must hold for all fixed $\theta, P\in \cH(\theta)$ and $b$. 
We may then think of $b$ as set in advance: if we set $b= 1/\alpha$ for some $0 < \alpha \leq 1$, then the requirement says that $\delta_b(Y)$ is a standard $(1-\alpha)$-confidence interval. Thus, under BIND, standard confidence intervals $\delta_b$  coincide with Type-I risk safe confidence intervals as defined above: just as for p-value based tests in Section~\ref{sec:gnp}, under BIND the new setting is simply  an equivalent reformulation of 
the existing theory of confidence sets.
Yet, again, if BIND is violated, then standard confidence intervals are not Type-I safe any more,
whereas e-based confidence intervals still are. 
 
\begin{example}{\bf [Ex.~\ref{ex:prebexample}, Continued]}\label{ex:normalcont}
Suppose you observe $Y=y$, $B=b$.  
Let us use the e-confidence intervals as defined relative to a particular anticipated $n^*$ and $\alpha^*$. Using \eqref{eq:normalrejectb} and substituting $1/b$ for $\alpha$ (so that now $c = (n^*/n) \cdot (\log (2 b)/(\log (2/\alpha^*))$), gives that 
$
\forall \theta \in \Theta: L_b(\theta,\delta_b(y)) \leq S_{\mean}
$ with $\delta_b(y) = [\theta_L,\theta_R]$ (i.e. compatibility (\eqref{eq:compatibleB}) holds  and hence Type-I risk safety \eqref{eq:above}
holds as well) as soon as   $\Lmean \leq  
\hat\mean -  A$ and
$\Rmean \geq \hat\mean + A$ where 
\begin{align}
\label{eq:aardrijkskunde}
A = \sqrt{\frac{2}{n}} \cdot \sqrt{
\log (2 b)} \cdot g(c).
\end{align}
We may choose $\delta_B(Y) = [\hat\theta-A,\hat\theta+A]$ to satisfy this with equality to make the interval as narrow as possible, making our interval admissible. 
We are then 
\commentout{
This formula is valid for $b  > 1/2$ and useful for $b \geq 1$. For, if $b < 1$ then the desired risk bound is obtained trivially by any interval, including the empty one: the maximum accepted risk is $1$ and the maximum loss, $b$, would then be smaller.}
guaranteed Type-I safety, \eqref{eq:above}, irrespective of the definition of $B$. In contrast, it is not clear how to construct Type-I safe CI's for data-dependent $B$ 
without e-values.
We might be tempted to do this based on {\em confidence distributions\/} (\cd's) \cite{Cox58,SchwederH16} that summarize confidence intervals for each $\alpha$ into a posterior-like quantity, or {\em objective Bayes posteriors} \cite{BergerB85}, but as we now show, this can have 
bad results.
\end{example}
\subsection{$\cd$'s and O'Bayes Posteriors are not valid Post-Hoc}\label{sec:obayes}
Consider the normal location family again. With the standard (uniform, improper) `objective Bayes' prior for this family and data $Y =x^n$, the posterior $W^{\circ} \mid Y= x^n$ has a normal density $w^{\circ}(\theta \mid x^n)$ with mean and median equal to the MLE $\hat\theta(x^n)$
and variance $1/n$ \citep{BergerB85}. In this case the objective Bayes posterior also coincides with the {\em fiducial\/} \cite{Hannig16} and the {\em confidence\/} distribution (\cd) \citep{SchwederH16} based on $x^n$, 
\commentout{
For 1-dimensional parametric families $\cP = \{P_{\theta}: \theta \in \Theta \}$, such \cd's are defined  as data-dependent distributions $W \mid Y= x^n$ on $\Theta$ (i.e., like `posteriors') with the property that for any $\alpha_1, \alpha_2 \geq 0$ with $0 \leq \alpha_1 + \alpha_2 \leq 1$, the interval $[\theta_L(Y),\theta_R(Y)]$ with $W(\theta \leq \theta_L(Y) \mid Y= x^n) = \alpha_1, W(\theta \geq \theta_R(Y) \mid Y= x^n) = \alpha_2$ is an $1-\alpha_1-\alpha_2$ confidence interval, thereby `summarizing' one- and two-sided confidence intervals for each $\alpha$ into a single distribution. For example,}
and has an exact coverage property: 
if we let $[\theta_L,\theta_R]$ represent the standard $(1-\alpha)$-Bayesian credible interval based on $w^{\circ}(\theta \mid x^n)$, i.e. taken symmetrically around the MLE $\hat\theta(x^n)$ then this coincides exactly with the standard $(1-\alpha)$-confidence interval, e.g. for $\alpha = 0.05$, we have $\theta_L = \hat\theta(x^n) - 1.96/\sqrt{n}$, $\theta_R = \hat\theta(x^n) + 1.96/\sqrt{n}$.  

The question is now how to base inferences on the objective Bayes or \cd's within our current GNP decision problem, i.e. if the goal is to come up with an interval as narrow as possible that contains the true $\theta$, where making a mistake is weighted by some $B$ that is determined post-hoc. 
Upon observing $Y=x^n$ and $B=b$, and based on the $\cd$ $w^{\circ}(\theta \mid Y)$, one would presumably pick the  smallest interval symmetric around $\hat\theta$ for which the Bayes posterior satisfies the required risk bound, i.e. $\delta'_b(Y) = [\hat\theta - A,\hat\theta +A]$ where $A$, depending on $b$, is the smallest number such that 
\begin{align}\label{eq:prebelow}
& {\bf E}_{\bar\theta \sim W^{\circ} \mid Y=x^n}[L_b(\bar\theta,\delta'_b(x^n))]  \leq 1, \\
    \label{eq:below}
\text{\ i.e.\ }  & {\bf E}_{\bar\theta \sim W^{\circ} \mid Y=x^n}[b \cdot {\bf 1}_{\bar\theta \not \in 
[\hat\theta-A,\hat\theta+A]} ] \leq 1,
\end{align}
with $b$  the observed value taken by $B$; and for this smallest $A$, \eqref{eq:below} holds with equality.  Since $W^{\circ}( \bar\theta \not \in 
[\hat\theta-A,\hat\theta+A]  \mid Y= x^n) = 1/b$, this $\delta'_b(Y)$ is equal to the standard $(1-\alpha)$-confidence interval for $\alpha = 1/b$.
The intuitive appeal for choosing this $\delta$ is clear: \eqref{eq:below} expresses that as a DM one can expect the loss given the data to be bounded by $1$; one simply wants to pick the smallest, most informative interval for which this holds true. Yet the {\em real\/} expectation of the loss  may very well be different from \eqref{eq:below} --- assuming that $B$ is a fixed function of $Y$, 
it is given by 
\begin{equation}
    \label{eq:therealmccoy}
{\bf E}_{Y \sim P_{\theta^*} } \left[B(Y) \cdot {\bf 1}_{\theta^* \not \in 
\delta'_{B(Y)}(Y)} \right],
\end{equation}
with $\theta^*$ indexing the true sampling distribution. This quantity may be much larger than $1$ if $B$ is dependent on $Y$. We provide a simple yet extreme example (inspired by the less extreme Example 8 of \cite{grunwald2018safe}) with $n=1$, $Y=X_1$ (equivalently, think of $Y$ as the Z-score corresponding to a larger sample): fix any $\epsilon > 0$. If, whenever $Y \geq \epsilon$, we set $B := 1/ (2 F_0(-Y + \epsilon^2/Y))$ where $F_0$ is the CDF of a standard normal, 
then, as demonstrated in the SI, under $\theta^*= 0$, \eqref{eq:therealmccoy} evaluates to $\infty$, irrespective of the definition of $B(Y)$ for $Y < \epsilon$. In particular we may set $B =1$ for such $Y$, corresponding to the decision problem being `called off', because the required bound  \eqref{eq:below} is then  achieved trivially by issuing the empty interval. 

\paragraph{Repercussions for Neyman's Inductive Behavior} This discrepancy between what one {\em believes\/} will happen according to a posterior (risk bounded by 1) and what actually will happen (potentially infinite risk) has repercussions for Neyman's interpretation of statistics as long-run performance guarantees of inductive behavior. To illustrate, imagine a DM who is confronted with such a decision problem many times (each time $j$ the underlying $\theta_{(j)}$ with  $Y_{(j)} \sim P_{\theta_{(j)}}$ and the sample size $n_{(j)}$ and the importance function $B_{(j)}$ may be different). Then, based on \eqref{eq:below} she might think to have, by the law of large numbers, the guarantee that, almost surely,
\begin{equation}
    \label{eq:inductivebehaviourci}
 \lim \sup_{m \rightarrow \infty}    \frac{1}{m} \sum_{j=1}^m  B_{(j)} \cdot {\bf 1}_{\theta \not \in 
[\theta_L(Y_{(j)}),\theta_R(Y_{(j)})]} \leq 1. 
\end{equation}
Unfortunately however, this statement is likely false if in reality there is a  dependence between $B_{(j)}$ and $Y_{(j)}$ 
In the SI  we show that, based on the example above, the average in \eqref{eq:inductivebehaviourci} may in fact a.s. converge to infinity, even though the individual $B_{(j)}$'s look pretty innocuous.
A first reaction may be to require the DM to address this problem by modeling the dependency between $B_{(j)}$ and $Y_{(j)}$. But the precise relation may be unknowable, and then it is not clear how to do this. To avoid the issue one may output e-based CIs  or, equivalently but perhaps more illuminatingly, CIs based on the {\em e-posterior\/} that we now introduce. 
\subsection{The E-Posterior remains valid Post-Hoc}
\label{sec:eposteriorsucceeds}
Let $\ecol = \{S_{\theta}: \theta \in \Theta \}$ be an e-collection. Just like it is  tempting to interpret a `system' of confidence intervals, one for each $\alpha$, i.e. a $\cd$, as a type of `posterior', one can also view the $S_{\theta}$-reciprocal $\bar{P}(\theta \mid Y) := S^{-1}_{\theta}(Y)$ as a type of `posterior representation of uncertainty'  for parameter $\theta$. This idea has been conceived of independently by \cite{waudby2023estimating} and  \cite{Grunwald23},  who called $\bar{P}(\theta \mid Y)$ the {\em e-posterior}. 
The crucial difference between  e-posteriors and  $\cd$s is that the former enable valid inferences under specific post-hoc, data-dependent assessments of Type-I risk, whereas standard \cd's can only be validly used as in \eqref{eq:prebelow} if BIND holds. We thus recommend e-posteriors, like Cox \cite{Cox58} 
did for \cd's, as a summary of estimation uncertainty --- but a summary that is significantly more robust than that provided by \cd's. 

Using the e-posterior we can  re-express compatibility, \eqref{eq:compatibleB}, as 
\begin{align}\label{eq:postposterior}
 \sup_{\theta \in \Theta} \bar{P}(\theta \mid y) \cdot  L_b(\theta,\delta_b(y)) \leq \ell,    
\end{align}
with conventions about $0\cdot \infty$ as underneath \eqref{eq:evaldr} and  $\ell = 1$. We already know that $\delta$ satisfying \eqref{eq:postposterior} are Type-I risk safe irrespective of how $B$ is defined. The rewrite suggests an analogy to the Bayes posterior risk assessment, \eqref{eq:prebelow}: if we replace objective Bayes/\cd-posterior {\em expectation\/} by e-posterior {\em maximum}, we get Type-I risk safety without the BIND assumption.
\commentout{Formally, let $Y, \cY$ be as before and let  $\ecol = \{S_{\theta}: \theta \in \Theta \}$ be a collection such that for each $\theta \in \Theta$, $S_{\theta}=S_\theta(Y)$ is an e-variable relative to null hypothesis $\cH(\theta)$.
The {\em e-posterior\/} (we leave out the word `distribution' since technically it isn't; rather it is a posterior quantification of uncertainty) corresponding to $\ecol$ is defined simply by setting, for all $y\in \cY$,  $\bar{P}(\theta \mid y) := S^{-1}_{\theta}(y)$, with conventions about division by $0$ as underneath \eqref{eq:evaldr}. 
We have already seen how to express compatibility, and hence Type-I risk safety, in terms of e-posteriors in \eqref{eq:compatibleB}. 
This suggests a daring generalization: what if we do not want to impose a strict Type-I risk bound $\leq 1$ as in \eqref{eq:above}, but instead want to assess the loss of any given decision rule  $\{\delta_b: b \in \cB \}$ with $\delta_b: \cY \rightarrow \cA_b$? \eqref{eq:compatibleB} suggests that, for given $Y=y$ and $B=b$, it is meaningful to asses our expected loss as 
$$
\sup_{\theta \in \Theta} \bar{P}(\theta \mid y) \cdot L_b(\theta,\delta_b(y)).
$$}

\cite{Grunwald23} shows that, for general bounds $\ell$ and with  $L_b$ replaced by general loss functions, without Type-I/II-dichotomies, assessment \eqref{eq:postposterior} is meaningful and provides a non-Bayesian alternative for Bayes-posterior expected loss assessment. 
In that paper, I also list a variety of e-posteriors, including an extension of the one of Example~\ref{ex:prebexample} to general exponential families, and  point out deeper relations between e-posteriors and Bayesian posteriors. 
\commentout{MAYBE, and introducing others which give wider e-confidence intervals at each fixed $n^*$ but work considerably better if the actual $n$ is very different from the anticipated $n$ (one can get widths of order $C \cdot \sqrt{(\log n)/n}$ uniformly in $n$ for a small constant $C$, and with the  techniques of  `stitching' \citep{howard2018uniformb} or `switching'  \citep{PasG18} even of order $\sqrt{(\log\log n) / n}$, yet this comes at the cost of a  significantly worse multiplicative constant). }  
In the present paper, we  merely present the e-posterior as a graphical tool which summarizes the e-based confidence intervals as given by \eqref{eq:compatibleB}
and helps to visualize how they relate to standard confidence intervals: see Figure~\ref{fig:eposterior}. 
\begin{figure}
\includegraphics[width=7.4cm]{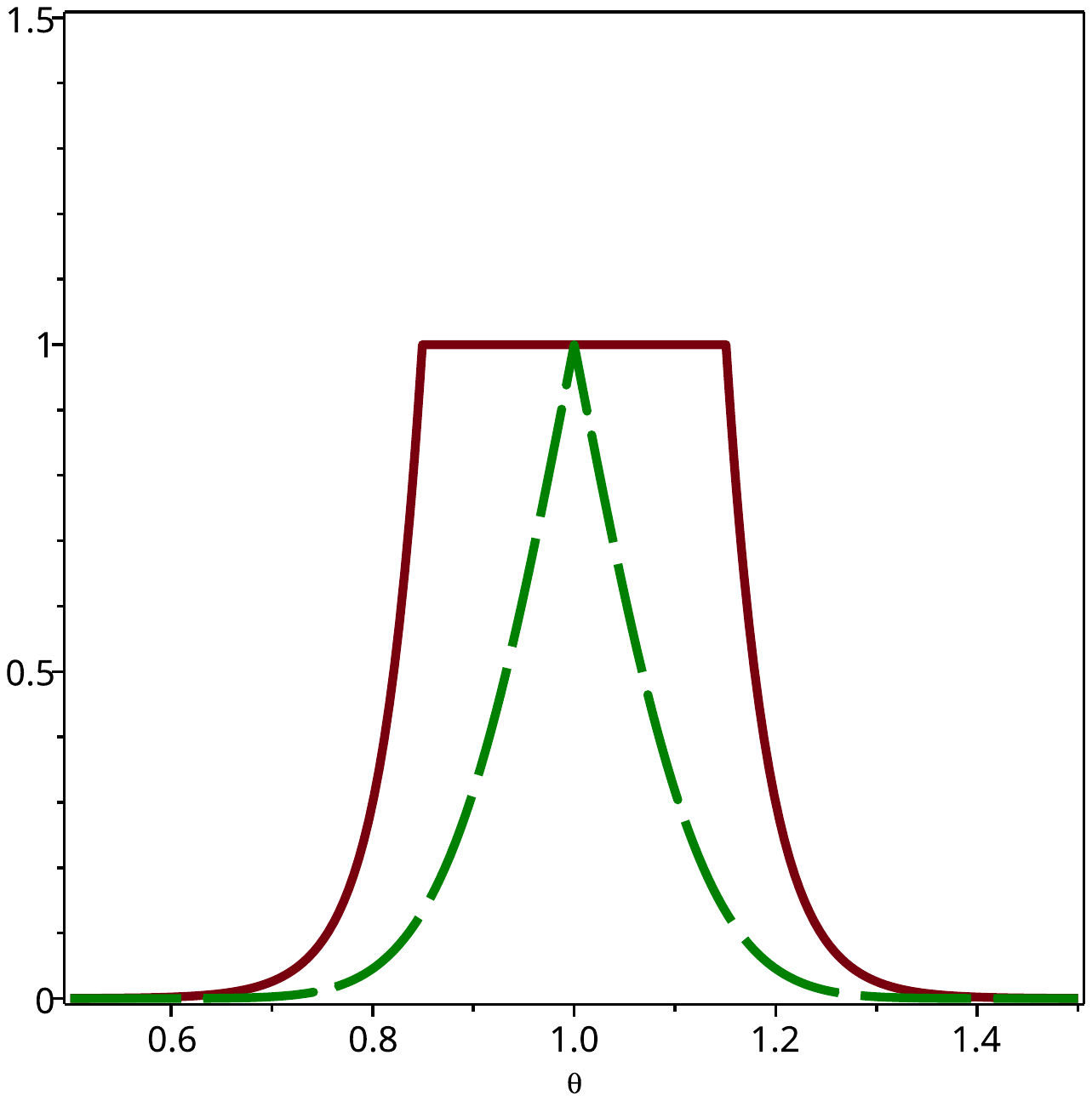} \hspace*{0 cm}
 \includegraphics[width=7.4cm]{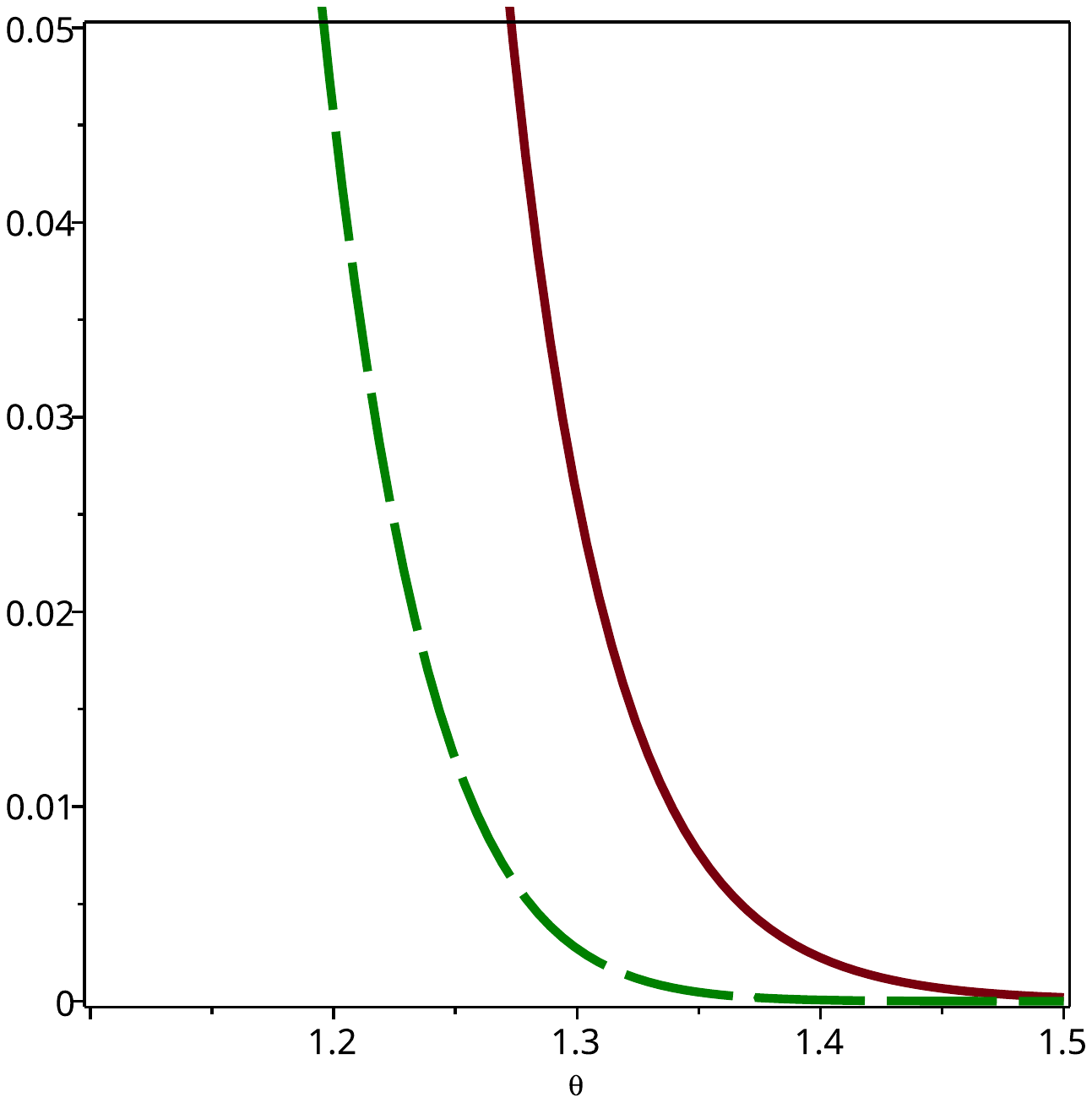} \vspace*{-3 cm}\ 
\caption{\label{fig:eposterior} The solid line depicts the e-posterior corresponding to the e-collection of Example~\ref{ex:prebexample} capped at 1, i.e. $\min\{1, \bar{P}(\theta \mid y)\}$, for data $y=x^n$ with $n=100$ and MLE $\hat\theta(y) =1.00$.  The dashed green line depicts twice the tail area of the objective Bayes posterior (\cd) $W^{\circ} \mid Y = x^n$ of \eqref{eq:below}, given by $f(\theta) := 2W^{\circ}(\bar\theta \geq \theta \mid y) = 2 \int_{\theta}^{\infty} w^{\circ}(\bar\theta \mid y) d \theta$.
The standard two-sided $(1-\alpha)$-confidence interval is given by $[\theta_L,\theta_R]$ where $\theta_L < \hat\theta = 1.00$ is the leftmost $\theta$ at which the dashed green curve takes value $\alpha$, and $\theta_R$ is the rightmost such $\theta$. The $(1-\alpha)$-e-confidence interval based on $S_{\theta}$ as defined above \eqref{eq:normalrejectb}, and with boundaries approximately equal to, \eqref{eq:normalrejectb}, is given by the $[\theta_L,\theta_R]$ at which the solid line takes value $\alpha$. The right picture zooms in on the right  for $\alpha \leq 0.05$. As expected, the dashed line hits $0.05$ for $\theta_R = \hat\theta(x^n) + 1.96/\sqrt{n} = 1.196$, the solid (e-based) line at $\theta_R = \hat\theta(x^n) + 2.72/\sqrt{n} = 1.272$.}
\end{figure}

\section{State of the Art}
\label{sec:stateoftheart}
The modern development of e-values and e-processes  started only  in 2019 when first versions of the four ground-breaking papers 
\cite{GrunwaldHK19B,VovkW21,Shafer19,wasserman2020universal}
appeared on arxiv.
Since then, development has been remarkably fast, often centering around GRO {\em (growth-rate optimal)\/} e-variables and -processes (see \cite{GrunwaldHK19B,ramdas2023savi} who also provide more historical context). Growth rate (also called {\em e-power\/} \cite{zhang2023exact}) is a natural analogue of power in the e-process setting, in which sample sizes are not fixed in advance, 
related to the minimum sample size needed to reject the null and \textsc{ci} width. 
As a first step, GRO e-methods were successfully developed for basic workhorses of statistics such as the z-test (the one appearing in Example~\ref{ex:prebexample} has GRO status), the t-test \cite{GrunwaldHK19B}, the test of two proportions \cite{TurnerG22B} and the logrank test \cite{TerschurePLG21},
 which has been successfully deployed in a `live' meta-analysis of ongoing clinical trials
\cite{terschure2022bcg}. 
The t-test setting has been extended to general tests and \textsc{CI}'s with group invariances and linear and nonparametric Gaussian process regresssion   \cite{perez2022estatistics,lindon22anytime,neiswanger2021uncertainty}.
The test of two proportions has been extended  to $k$-sample tests of general exponential families \cite{HaoGLLA24},
CMH tests  \cite{TurnerG23} and to conditional independence (of $X$ and $Y$ given $Z$) testing under a {\em model-X assumption\/} 
(combineable with arbitrary models for $Y\mid X,Z$) 
\cite{GrunwaldHL22}.
E-variables that 
are not GRO 
yet still have good power-like properties have been quite successful in multiple testing applications \cite{wang2022false,wang2022values} as well as several other nonparametric problems  \cite{Podkopaev22skit,waudby2023estimating,henzi2021valid}
 --- the recent overview \cite{ramdas2023savi} provides a comprehensive list.  
In all applications mentioned, the qualitative behavior is similar to that of 
Example~\ref{ex:prebexample}, with \textsc{ci} widths of order $O(f(n)/\sqrt{n})$ per parameter of interest, $f(n)$ varying from constant to $O(\sqrt{\log n})$, depending on the specific
application. 

\paragraph{(Current) Limitations and Challenges} These initial successes notwithstanding, the development of e-values is, of course, still in its infancy, competing with almost a century of p-value development. As such, many challenges remain. To appreciate these, we first note that the aforementioned GRO-type approaches can in principle be made competitive, in terms of sample sizes needed to draw a conclusion, with classical ones that rely on BIND --- see below; sometimes they even significantly beat such classical methods (e.g. \cite{TurnerG23,waudby2023estimating}). Also, \cite{GrunwaldHK19B}
shows that GRO e-values exist and can be calculated for very general testing problems.
\commentout{
main ones, we need to say a little about e-value design. Currently, broadly speaking, two dominant methods for constructing e-values for particular settings have emerged: {\em universal inference\/} and {\em GRO}. Both are surprisingly generally applicable --- the first essentially in arbitrary parametric and nonparametric testing settings, the second under some very weak regularity conditions on the model(s) involved. However, they each have their own weakness: the first leads to procedures that have less power than their classical counterparts, and, worse, the loss of power increases with the dimensionality of the problem CITE. The second (which is indeed also known as `optimizing for e-power' CITE) behaves better power-wise, but it is in general}
Yet in general, this calculation is not efficient.
In some cases (such as \cite{perez2022estatistics,TurnerG23,GrunwaldHL22} mentioned above), they admit an analytical and hence efficiently calculable expression, but for others,  they do not. 
These hard cases include regression (i.e. $Y=f_{\theta}(X,Z) + \textsc{noise}$) that involves a nonlinearity, such as GLMs and  Cox proportional hazards, whenever the variable $X$ to be tested (e.g. treatment vs. control) does {\em not\/} satisfy the model-X assumption (conditional distribution of $X$ given $Z$ known). While model-X  is automatically satisfied in clinical trials, 
there are of course many important cases in which it is not. {\em Universal inference\/} \cite{wasserman2020universal} provides an alternative generic e-design method that does lead to efficiently calculable e-values in such cases, but in regression problems it is not competitive in terms of power with classical methods for medium- to high-dimensional models \cite{tse2022note} --- its strength has rather been to provide  e-values for complex $\cH(\snull)$ that have simply eluded classical testing \cite{dunn2021universal}. 

\paragraph{Challenges -- II} With GRO-type methods one can obtain comparable performance in terms of power
as compared to
classical approaches. In many (not all) settings though, one needs to engage in optional stopping to achieve this. For a broad class of e-values, this is no problem (\cite{GrunwaldHK19B} provides a detailed analysis): all coverage and Type-I risk guarantees are retained under such optional stopping. Still, it points towards a second challenge for e-methods, sociological/psychological rather than statistical: it requires researchers to think differently, and this is, of course, always difficult to accomplish. In this respect, the tech industry is at the forefront: anytime-valid methods based on e-processes have been adopted by several major tech companies \cite{lindon22anytime}.
\commentout{As a case in point, the present paper can be read as implying that a small loss in power (broader confidence intervals) should be quite acceptable if it leads to validity under a much broader class of loss functions.  Still, when I present this work to other researchers, the first reaction is almost always `but what about power?'}

\section{Discussion, Future Work, Conclusion}
\label{sec:discussion}
We provide a few concluding remarks. First we 
analyze  in what sense we solved the `roving $\alpha$' issue that motivated this work. Second, we discuss related work. Third, we suggest a `road ahead' for e-methods. 
\paragraph{Roving $\alpha$ Revisited: the Quasi-Conditional Paradigm} 
Assume we have a prior on  $\Hnull$ and $\Halt$ and priors $W_0$ and $W_1$ on the distributions inside these hypotheses. We can then use Bayes' theorem to calculate the Bayes posterior $P(\Hnull \mid Y)$ based on data $Y$. Suppose we reject the null if $P(\Hnull \mid Y) \leq \alpha$. We may  then define, for all $y$ for which this holds, i.e. for which we reject the null,  the 
{\em conditional Type-I error probability\/}  $\hat\alpha$ to be simply equal to this posterior probability, $\hat\alpha:= \hat\alpha_{|y} := P(H_0 \mid Y=y)$. This implies that, for any fixed $\alpha_0\leq \alpha$, for any long sequence of studies, with probability tending to one, 
\begin{align}\label{eq:calibrate}
\parbox{20em}{``among all studies with $\hat\alpha \leq \alpha_0$, we make a \\ Type-I error  at most  a fraction $\alpha_0$ of the time''.}
\end{align}
Such a fully conditional statement, with post-hoc determined $\hat\alpha_{|Y}$,  is  only correct if the priors can be fully trusted, i.e. if one accepts a fully subjective Bayesian stance. It would definitely be incorrect if we set $\hat\alpha_{|Y}$ either to a p-value or the reciprocal of an e-value based on $Y$. 
Still, as we have seen, if we instead use e-values to perform a data-dependent action, which is allowed to get more extreme (higher Type-I loss) as our evidence against the null increases (higher e-value) according to the maximally compatible rule (which in simple cases is given by \eqref{eq:evaldr}), then we {\em do\/} get an `unconditionally' valid bound on Type-I risk. Thus, using e-values, setting a roving $\alpha$ to be equal to $\hat\alpha:= \ell/S(y)$ for the observed $y$ is still incorrect if we interpret it as expressing \eqref{eq:calibrate}; but it is correct if we interpret it as setting a `roving bound' of $\ell/\hat\alpha$ on the Type-I loss $L_B(\snull,a)$ we dare to make: 
if we make sure to pick $a$ so that $L_B(\snull,a) \leq \ell/\hat\alpha$, then we have compatibility and hence Type-I risk safety, \eqref{eq:TypeIRiskSafeSimplified}.
Note that $B$ is allowed to be  any function of,  hence `conditional on'  data; but its performance  is evaluated `unconditionally', i.e. by means of \eqref{eq:TypeIRiskSafeSimplified} which is an unconditional expectation. This {\em quasi-conditional stance}, explained further in \cite{Grunwald23}, provides a middle ground between fully Bayesian and traditional Neyman-Pearson-Wald type methods and analysis. 

\paragraph{Where does the Type-I risk bound $\ell$ come from?}
Whereas $B$ may arbitrarily depend on data  $Y$, the upper bound $\ell$ in \eqref{eq:TypeIrisksafe} has to be set independently of $Y$ after all. 
It may still vary from decision problem to decision problem though (in the SI Appendix we explain what this means in terms of Neyman's inductive behavior paradigm and we explain that setting $\ell=1$, as was done for mathematical convenience in Section~\ref{sec:math}, is unproblematic). In many practical testing problems, we might expect that for all $b \in \cB$, $\cA_b$ contains a special action $0$, which stands for `do nothing' (keep status quo),  which would then have the same Type-II loss under all $b \in \cB$, i.e. there is an $\ell'$ such that for all $b \in \cA_b$, $L_b(\alt,0)= \ell'$. We might then simply set $\ell = \ell'$, making sure that we can expect our result  (with all costs and benefits incorporated), whatever action we take, to be no worse than ``the cost of doing nothing when  we really should have done something''. 
\paragraph{Related Work: Inferential Models}
Like we do in Example~\ref{ex:prebexample}, 
Martin, Liu and collaborators  
\citep{martin2015inferentialB,Martin21}  point out discrepancies between what one would expect to be a valid confidence set  according to a fiducial, \cd\  or Bayesian posterior and what are actually valid confidence sets according to the unknown, true distribution. They provide {\em inferential models (IMs) \/} as a safer alternative. Unlike e-posteriors, the specific IMs proposed by \citep{martin2015inferentialB} still work under the BIND assumption and thus will not provide reliable inferences if BIND does not hold. But it may very well be that some other  IMs (IMs constitute a family of methods, not a single method)  essentially behave like e-posteriors.
\commentout{, as well as placing both IMs and the present work in the context of {\em safe probability\/} \cite{grunwald2018safe}, a method for assessing precisely what decision tasks an inference method can be safely used for and what tasks it cannot.}

\paragraph{The Road Ahead}
Future work 
will include a further investigation of the `quasi-conditional' idea launched above, as well as of the precise relation to Martin's IMs and other related uses of e-variables  such as \cite{BatesJSS22} who, like us, employ e-values with a Type-I/II-error distinction with more than 2 actions.

Another unresolved fundamental issue is this: most practitioners still interpret p-values in a Fisherian way, as a notion of {\em evidence\/} against $\cH(\snull)$. Although this interpretation has always been controversial, it is to some extent, and with caveats (such as `single isolated small p-value does not give substantial evidence' \cite{mayo2018statistical} or `only work with special, {\em evidential\/} p-values \cite{Greenland22}'), adopted by highly accomplished statisticians, including the late Sir David Cox \cite{Cox18,mayo2006frequentist}. Even Neyman \cite{Neyman76} has written `my own preferred substitute for `do not reject $H$' is `no evidence against $H$ is found'. In light of the present results, one may ask if, perhaps, {\em e-values are more suitable than p-values\/} as such a measure. We preliminarily conjecture they are, and motivate this in the SI  --- although a proper analysis of such a claim warrants  a separate paper, which we hope to provide in the future.

Perhaps more important for practice than all of this though, in light of Section~\ref{sec:stateoftheart} above,  is the further development of practically useful e-variables for standard settings (such as GLMs) in which they are not yet available, as well as more accompanying software such as  \cite{ly2022safestats}. 
%

\paragraph{Conclusion: A different kind of Robustness}
Standard $\textsc{p}$ and $\textsc{cs}$-based decision rely on  BIND, an assumption that will often be false or unverifiable at the time study results are published. In this paper we showed that e-values provide valid error and risk guarantees without making such assumptions, and are therefore {\em robust\/} tools for inference. But whereas `robustness' usually refers to robust inference in the presence of outliers, or model structure or noise process misspecification, this is a different, much less studied form of robustness: robustness in terms of the actual decision task that the study results will be used to solve. 

\DeclareRobustCommand{\VOORVOEGSEL}[3]{#3} 
\paragraph{Acknowledgements} The author would like to thank an anonymous referee and Dr. W.~Koolen, who both independently alerted him to the fact that, without essential loss of generality, one may assume Type-II loss to decrease whenever Type-I loss increases. 

\bibliography{master,peter,local}

\newpage

\appendix
\section*{\huge Supporting Information Appendix}
\section{Supporting Information for Section~\ref{sec:gnp} and Section~\ref{sec:math}.\ref{sec:TypeI}}
\paragraph{Unbounded expected loss based on  \eqref{eq:pvalb} and an `improvement' of \eqref{eq:pvalb}}
This issue is best illustrated by (but certainly not limited to)  a discrete-valued p-value $\pval$ that can take values $1,1/2,1/4,1/8,\ldots, 1/2^k$ for some $k > 0$ and that is piece-wise strict, i.e. it  satisfies $P_0(\pval \leq \alpha) = \alpha$ for $\alpha \in \{1,1/2, \ldots, 1/2^k \}$.  Consider a GNP decision task as in Section~\ref{sec:TypeI} with loss function satisfying $L_b(\snull,a) = 2 a$, for $a \in \cA_b =  \{0,1 \ell , 2 \ell, 4 \ell, \ldots,2^k \ell \}$. 
Based on \eqref{eq:pvalb}, upon observing $\pval = 2^{-c}$, one would take action $2^c$. The resulting expected loss, analogously to \eqref{eq:losstoolarge}, is given by $\sum_{c=1}^k 2 \cdot 2^{-c} 2^{c} = 2 k$ which goes to $\infty$ as we make $k$ larger --- showing that the expected loss can be unbounded if we base decisions on \eqref{eq:pvalb}. Now, instead of using \eqref{eq:pvalb} it may seem more reasonable to pick the largest $a$ such that 
\begin{equation}\label{eq:pvalc}
\textsc{q}(y) \cdot L_b(\snull,a) \leq \ell,
\end{equation}
where $\textsc{q}(y) = \pval(y)/2$: with this modification, for each $a\in \cA_b$, we end up multiplying $L_b(\snull,a)$ in \eqref{eq:pvalc} with exactly the probability that $a$ will be selected (rather than, as in \eqref{eq:pvalb}, with a larger probability). For example, $a= 2^{c}$ will be selected if $\textsc{q}(y) = 2^{-c}$; this happens iff $\pval(y) = 2^{-c+1}$, which happens with exactly probability $2^{-c}$, so with probability $\textsc{q}(y)$).  
Yet still, using \eqref{eq:pvalc} leads to unbounded expected loss: in the above sample the expected loss is now $k$ rather than $2k$, still growing linearly in $k$. 

\paragraph{Standard conversions of p-values into e-values are sub-optimal}
\cite{shafer2011test,ShaferV19,VovkW21} have studied functions that convert arbitrary p-values to e-variables.
These {\em calibrators\/} are strictly decreasing functions $f$, such that  $S(Y) :=f(\pval(Y))$ is an e-variable whenever $\pval$ is a p-value. Calibrators invariably have the property that as $\pval \downarrow 0$, $f(\pval)$ grows towards $\infty$ more slowly than $1/\pval(Y)$ (note that, for any e-variable $S$, $1/S$ is a (conservative) p-value but the converse does not hold).  For example, \cite{ShaferV19},
$f(\pval) = 1/\sqrt{\pval} -1$ is a calibrator.
Given that such calibrators exist, one might wonder if in this paper we are really merely advocating to change the {\em scale\/} at which evidence against the null is expressed: isn't it sufficient to take a p-value to express evidence, convert it to an e-value, and then use the maximally compatible decision rule \eqref{eq:evaldr}? If so, this would undercut our arguments for using e vs. p. The answer is that, while calibrating p-values and then using \eqref{eq:evaldr} does give us Type-I risk safety, it is still not advisable because e-variables arising from calibrated p-values are typically far from optimal. Intuitively, a `good' e-variable, relative to a given alternative, is one that tends to be large (provide much evidence) if the alternative is true.
This can be formalized in terms of power or, as stated in the main text, GRO. For example, suppose we aim to test null $\theta \leq 0$ against alternative $\theta \geq \delta$. If we take the 1-sided e-variables $S^+_{\theta}$ for the normal location family as defined underneath \eqref{eq:nstarfun} with $\alpha^*=\alpha$ and $n^*=n$ specified correctly, then to get power $0.8$ we need a factor of $\approx 1.75$ more data points than if we use the standard UMP NP test (this follows from the derivation in \cite[Appendix B.6]{GrunwaldHK19B}; as explained there, the factor can be significantly reduced by optional stopping). If, instead, we use the p-value corresponding to this UMP test directly and turn it into an e-value by the above calibrator, we need a factor of $\approx 3.0$ more data \cite[Section 7]{GrunwaldHK19B}. 
The reason for this discrepancy is that calibrators work for arbitrary p-values and are  thus `blind' to the underlying sampling model (in this case, normal location). In order to get high power it is invariably (much) better to use e-values designed for the underlying model directly --- so, importantly, the distinction between e and p is not just a matter of scale and designing `good' e-values (with good GRO properties) is a nontrivial taks --- we cannot just take any given p-value with good power-properties and calibrate.   

\paragraph{Why normalizing $\ell$ to 1 in \eqref{eq:TypeIRiskSafeSimplified} and \eqref{eq:TypeIRiskSafeSimplifiedB} is not harmful}
This is further discussed, in a wider context, in the Supporting Information for Section~\ref{sec:discussion} further below. 
\section{Supporting Information for Section~\ref{sec:math}.\ref{sec:TypeII} until and including Section~\ref{sec:eposterior}.\ref{sec:retrace}}\label{app:proofs}
In this section we state and prove Theorem~\ref{thm:maingeneral}, a general result which has Theorem~\ref{thm:mainsimple} of Section~\ref{sec:math} as a special case, extending it to the case of general GNP decision problems (e.g. confidence intervals) as defined in Definition~\ref{def:gnpgeneral}, Section~\ref{sec:eposterior}. 
However, in the first subsection below, we first state and prove a simpler form of the theorem which works for {\em simple\/} GNP decision problems, as defined in Section~\ref{sec:math}, simplified even further by requiring countable $\cY$. This allows us to strip away all issues about `almost surely, measurability' etc. and focus on the key idea, which lies in the fundamental Lemma~\ref{lem:equalizersimple}. 
The proof of the general Theorem~\ref{thm:maingeneral} is based on essentially the same key insight but requires substantial additional notations and quantifications. We prepare these in 
Section~\ref{app:proofs}.\ref{app:definitions}
below, where we also explain why the probabilities in \eqref{eq:TypeIIbettera}, \eqref{eq:TypeIIbetterb} (used to define Type-II strictly-better-than and admissibility) as well as the expectation \eqref{eq:TypeIRiskSafeSimplified} may be undefined 
in pathological cases, and we extend the definitions of Type-II-strictly better and admissibility to ensure that these are always well-defined. 
We then 
state and prove Lemma~\ref{lem:equalizer}, the general version of Lemma~\ref{lem:equalizersimple} in Section~\ref{app:proofs}.\ref{app:lemma}, and state and prove the general theorem in Section~\ref{app:proofs}.\ref{app:finalresult}. 
Before we do all this though, we explain, as promised underneath Theorem~\ref{thm:mainsimple} in the main text, how we can make any GNP testing problem rich by adding a single additional loss function and why this makes the condition of richness a reasonable one. This is illustrated by Example~\ref{ex:add} which indicates why a condition like richness is necessary and thereby gives a (very) high-level intuition to the proof.
\paragraph{Enforcing Richness relative to $S$: why richness is a weak condition, and why it is needed}
For any sharp e-variable $S$ which we might want to base our decisions on, we can trivially {\em enlarge\/} any given
GNP testing problem 
$(\cB',\{(\cA_b, L_b(\snull,\cdot):  \cA_b \rightarrow \reals^+_0): b \in  \cB' \})$
by setting $\cB := \cB' \cup \{ \textsc{id}(S)   \}$, adding a loss function indexed by $\textsc{id}(S)$ with action space $\cA_{\textsc{id}(S)}$ set equal to the co-domain of $S$ and, for $s \in \cA_{\textsc{id}(S)}$, we set  
$L_{\textsc{id}(S)}(\snull,s):= s$. Then the extended GNP decision problem will automatically be rich relative $S$, and Part 2 of Theorem~\ref{thm:mainsimple} can be applied. In reality, DM usually is not aware of the full details of the problem anyway, being only presented one particular loss function $L_b$, an element of a set  $\{L_b: b \in \cB'\}$ that is unknown: DM will only know the definition of the particular function $L_b$ that she is presented with. 
Thus, assuming the set already contains this additional, special loss $L_{\textsc{id}(S)}$ does not really impose any additional condition on the DM and only serves to make the analysis more robust, it thus seems a reasonable assumption. It gives an imagined adversary who chooses $b=B(Y)$ as function of $Y$ more power, and as illustrated by the example below, without something like this added power, the theorem simply cannot hold. As such is analogous to (but not the same as) allowing an adversary to randomize between actions, as required for the minimax theorem in game theory. To take the analogy to the minimax theorem even further, we note that, just like in that theorem, one side of the proof (Part 1) is in essence trivial whereas the other side (Part 2) requires a sophisticated argument.  
\begin{example}
\label{ex:add}    
Consider a GNP testing problem and e-variable $S$ defined as follows:
\begin{enumerate}
    \item $\cY = \{0,10,20\}$; $\cH(\snull) = \{P_0\}$ with $P_0(Y=10) = 1/20$ and $P_0(Y=20)=1/40$. 
    \item $\cB = \{b_1\}$, $\cA_{b_1} = \{0,9, 19, 21\}$,  and for all $a \in \cA_{b_1}$, we set  $L_{b_1}(\snull,a) = a$.
    \item $S(y) := y$.
\end{enumerate}
We note that $S$ is a sharp e-variable, but the GNP testing problem is {\em not\/} rich relative to $S$. 
And indeed, the conclusion of Part 2 of Theorem~\ref{thm:mainsimple} is violated here: 
$$
\delta_{b_1}(0)=0 ; \delta_{b_1}(10) =9 ; \delta_{b_1}(20) =19 
$$
is seen to be a decision rule that is  maximally compatible with $S$, but it is not admissible: the decision rule 
$$
\delta'_{b_1}(0)=0 ; \delta'_{b_1}(10) =9 ; \delta'_{b_1}(20) = 21 
$$
is Type-II strictly better than $\delta$ yet still Type-I risk safe, since ${\bf E}_{P_0}[L_{b_1}(\snull,\delta'_{b_1}(Y))] = 9/20+ 21/40 = 39/40 < 1$. So we have a decision rule that is maximally compatible relative to a sharp e-variable yet not admissible --- this shows some additional condition such as richness is necessary. To get an initial idea why `richness' does the trick, 
let's enlarge $\cB$ as above to make the resulting GNP testing problem rich relative to $S$. That is, we add loss function indexed by $b_2 := \textsc{id}(S)$, so that $\cA_{b_2} = \{0,10, 20\}$, and for all $a \in \cA_{b_2}$, we have  $L_{b_2}(\snull,a) = a$. Decision rule $\delta$ above was maximally compatible in the original problem, and the only way to extend it to the enlarged problem while keeping it maximally compatible is to set 
$\delta_{b_2}(y) = y$ for all $y \in \cY$. But now, in this enlarged problem, $\delta$ has become admissible! 
Rather than proving this in full generality we will just show that the decision rule $\delta'$ above that witnessed $\delta$'s inadmissibility in the original problem will not witness it any more in the enlarged problem. To see why, note that to witness inadmissibility of $\delta$, we must have that $\delta'$ is Type-II strictly better than $\delta$ and at the same time Type-I risk safe. The only way to extend $\delta'$ to the enlarged problem while keeping it strictly better than $\delta$ is to set it such that $\delta'_{b_2}(y) \geq \delta_{b_2}(y)$ for all $y \in \cY$. But then it is not Type-I risk safe any more, so this extended $\delta'$ does not show $\delta$ to be inadmissible! To see why the extended $\delta'$ is not Type-I risk safe any more, note that, by adding loss $L_{b_2}$, we gave the imagined adversary more power: upon observing $y=10$, the adversary can now choose $B=b_2$, and upon observing $y=20$, she can choose $B=b_1$. Then ${\bf E}_{P_0}[L_{B}(\snull,\delta'_{B}(Y)) \geq  (1/20) \cdot 10 + (1/40) \cdot 21 = 21/40 > 1$, so $\delta'$ is not Type-I risk safe. Lemma~\ref{lem:equalizersimple} and its generalization Lemma~\ref{lem:equalizer} further below formalize this idea and are the key to proving Part 2 of Theorem~\ref{thm:mainsimple} in the main text and its generalization Theorem~\ref{thm:maingeneral} below. 
\end{example}

\subsection{Theorem~\ref{thm:mainsimple} for countable $\cY$ and $\cH_(\snull)$ with full support}\label{app:supersimple}
Throughout this subsection we assume that we deal with a  GNP testing problem that is simple in the sense of Section~\ref{sec:math}, with countable $\cY$ and $\cH(\snull)$ with full support. `Full support' simply means that for all $y \in \cY$, all $P \in \cH(\snull)$, we have $P(Y=y) > 0$.  Because the testing problem is simple, we can define maximum compatibility in terms of \eqref{eq:maxcompatible}. 

So, fix any simple GNP testing problem of this type. For any given random variables $U= u(Y), V=v(Y)$, we write $U \leq V$ as an abbreviation of: for all $y \in \cY$, $u(y) \leq v(y)$; similarly for $U < V$ and $U=V$. 
\paragraph{Key Concept and Lemma: Equalizing Maximal Compatibility}
For any function $B: \cY \rightarrow \cB$, we say that  decision rule $\delta^{\circ}$ is {\em equalizing-maximally compatible\/} relative to e-variable $S$ {\em when restricted to $B$\/} if
\begin{equation}\label{eq:equalizersimple}
\text{for all $y \in {\cal Y}$}:\ 
L_{B(y)}(\snull,\delta^{\circ}_{B(y)}(y)) =  S(y), \text{\ i.e.\ }
L_B(\snull,\delta^{\circ}_B) = S.
\end{equation}
The following lemma is the key insight needed to prove Theorem~\ref{thm:mainsupersimple} below and hence the special case of Theorem~\ref{thm:mainsimple} in the main text with a simple GNP testing problem and countable $\cY$. It will later be generalized to Lemma~\ref{lem:equalizer} which plays the same key role for the general result Theorem~\ref{thm:maingeneral}.
  \begin{lemma}\label{lem:equalizersimple}
Fix any simple GNP testing problem as above. Suppose that $S$ is a sharp e-variable and  $\delta^{\circ}$ is a Type-I risk safe decision rule such that
there exists a function  $B: \cY \rightarrow \cB$ so that
$\delta^{\circ}$ is equalizing-maximally compatible  relative to $S$ when restricted to $B$, as above.  Then $\delta^{\circ}$ is fully compatible with $S$, i.e. for  all $b \in \cB$, we have: $L_b(\snull,\delta^{\circ}_b) \leq  S$. 
\end{lemma}
To understand the lemma, 
suppose we are given some e-variable $S$ and some Type-I risk safe $\delta$. 
Then $\delta$ need not be compatible with $S$ (it must be compatible with {\em some\/} $S'$, but not necessarily {\em this\/} $S$). The lemma says that if $\delta$ is in some specific sense `partially' compatible with $S$ though, namely for a specific $B$, and $S$ is sharp, then it must be fully compatible with $S$ after all. Now, in the case where the GNP testing problem has been made rich relative to $S$ by adding the special loss function $\textsc{id}(S)$ above, we would typically apply this lemma with $B=\textsc{id}(S)$, i.e. $B$ is a constant, independent of $Y$; but the lemma works even if $B$ may vary with $Y$.  
The surprising thing here is that compatibility relative to $B$ (which may even be the constant $\textsc{id}(S)$) has repercussions for the behavior of $\delta^{\circ}_{b'}$ for {\em all\/} $b' \in \cB$.

The lemma immediately leads to the following corollary: 
\newtheorem{corollary}{Corollary}
\begin{corollary}
Fix any simple GNP testing problem as above. If a decision rule $\delta^*$ is maximally compatible relative to a sharp e-variable $S$ (i.e. \eqref{eq:maxcompatible} holds) and equalizing-maximally compatible relative to the same $S$ when restricted to some function
$B: \cY \rightarrow \cB$ 
then any  $\delta^{\circ}$ that is equalizing-maximally compatible relative to $S$ when restricted to $B$ and Type-I risk safe must also be fully compatible with $S$ and hence, since $\delta^*$ is maximally compatible, satisfy $\delta^{\circ}_b \leq \delta^*_b$ for all $b \in \cB$.
\end{corollary}
\begin{proof}{\bf \ [of Lemma~\ref{lem:equalizer}]}
By Proposition~\ref{prop:typeIsafety}, there must be some e-variable $S'$ such that $\delta^{\circ}$ is compatible with $S'$, i.e. 
\begin{equation}\label{eq:ercsimple}
\text{for all $b\in \cB$:}\   
L_b(\snull,\delta^{\circ}_b) \leq  S'.
\end{equation}
By equalizing-maximal compatibility relative to $S$ when restricted to $B$, transitivity, and weakening \eqref{eq:ercsimple}, we must therefore also have
\begin{equation*}
S =  L_B(\snull,\delta^{\circ}_B) \leq  S',
\end{equation*}
so that $S \leq  S'$.
Suppose by means of contradiction that, even stronger, there is $y \in \cY$ such that $S(y) <  S'(y)$.
We know that for some $P_0 \in \cH(\snull)$,  ${\bf E}_{P_0}[S]=1$. But then ${\bf E}_{P_0}[S'] > 1$, so $S'$ is not an e-variable and we have arrived at a contradiction. 
Since we already established $S \leq  S'$  it follows that $S = S'$. But then using the inequality in   \eqref{eq:ercsimple} shows that for all $b \in \cB$, we have $L_b(\snull,\delta_b^{\circ}) \leq  S$ and the lemma is proved. 
\end{proof}
Armed with this result, we can now state and prove a restricted version of Theorem~\ref{thm:mainsimple}. 
\begin{theorem}\label{thm:mainsupersimple}
Consider any simple GNP testing problem as above, with countable $\cY$ and all $P \in \cH(\snull)$ having full support on $\cY$.  
\begin{enumerate} 
\item 
Suppose that decision rule $\delta$ is admissible. Then there exists an  e-variable $S$ such that $\delta$ is maximally compatible with $S$.
\item Suppose that $S$ is a sharp e-variable $S$  and $\delta^*$ is maximally compatible relative to $S$ (such a $\delta^*$ exists because we assume the GNP testing is simple); and assume further that the GNP testing problem is rich relative to $S$. Then $\delta^*$ is admissible. 
\end{enumerate}
\end{theorem}
\  \\
\begin{proof}{\bf \ [of Theorem~\ref{thm:mainsupersimple}]}

{\em Part 1}. 
Suppose that $\delta$ is admissible. Then $\delta$ is by definition Type-I risk safe. By Proposition~\ref{prop:typeIsafety} there must be an e-variable $S$ such that $\delta$ is compatible with $S$. Since $\delta$ is admissible, every  strictly better $\delta'$ is not Type-I risk safe, hence cannot be compatible with any e-variable; in particular, $\delta'$ is not compatible with $S$. Hence there exists no $\delta'$ that is strictly better than $\delta$ and compatible with $S$. It follows that $\delta$ is maximally compatible with $S$. 

{\em Part 2}. 
Let $\delta^*$ be maximally compatible  and let $\delta^{\circ}$ be another  Type-I risk safe decision rule. 
We will show that  $\delta^{\circ}$ cannot be Type-II strictly better than $\delta^*$; this implies the result. 

By our $S$-relative richness assumption and the construction of $\delta^*$, we know that there exists a function $B: \cY \rightarrow \cB$ with:
\begin{align}\label{eq:strasimple}
   L_{B}(\snull,\delta^*_{B}) = S.
\end{align}
We may assume that $\delta^\circ$ satisfies $\delta_{B}^{*} \leq \delta_{B}^{\circ}$, otherwise we already know that $\delta^{\circ}$ is not Type-II strictly-better. 
Now suppose by means of contradiction that for some $y \in \cY$, we have $\delta_{B(y)}^*(y)< \delta_{B(y)}^{\circ}(y)$. We then have for the $P_0 \in \cH(\snull)$ with ${\bf E}_{P_0}[S]=1$ (which exists by sharpness) that 
$$ 
1=  {\bf E}_{P_0}[S] = {\bf E}_{P_0}[L_B(\snull,\delta^*_{B})]
< {\bf E}_{P_0}[L_B(\snull,\delta^{\circ}_{B})],
$$
contradicting our assumption that $\delta^{\circ}$ is Type-I risk safe. 
We may thus assume $\delta_{B}^{\circ} =\delta_{B}^*$.

The corollary of Lemma~\ref{lem:equalizersimple} above now implies that $\delta_b^{\circ} \leq \delta^*_b$ for {\em all\/} $b \in\cB$ (i.e. not just for $B$), hence $\delta^{\circ}$ is not Type-II strictly better than $\delta^*$; the theorem is proved. 
\end{proof}

\subsection{Preparing the General Proof of Theorem~\ref{thm:mainsimple} and~\ref{thm:maingeneral}}
\label{app:definitions}
\newcommand{\w}{\ensuremath{\textsc{L}}}
\paragraph{Almost sure inequality}
Fix any GNP decision problem with parameter set $\Theta$, as in the general Definition~\ref{def:gnpgeneral}. In particular in some applications we may have  $\Theta= \{\snull\}$, then we really deal with a GNP testing problem and the notation $\leq_{\theta}$ that we will now define can in such cases be replaced by $\leq_{\snull}$. 

For all $\theta \in \Theta$, for functions $U,V: \cY \rightarrow \reals ^+_0$ we define
\begin{equation}\label{eq:weaka}
U(Y) \leq_{\theta} V(Y)
\end{equation}
to mean that for all $P \in \cH(\theta)$, for all $\epsilon > 0$ and every  measurable set $\cE \subset \cY$
such that for all $y \in \cE, U(Y) > V(Y) + \epsilon$, we have $P(\cE) = 0$.
Similarly, 
\begin{equation}\label{eq:weakb}
U(Y) <_{\theta} V(Y)
\end{equation}
is defined to mean that $U(Y) \leq_{\theta} V(Y)$ and there exist $P \in \cH(\theta)$,  $\epsilon > 0$ and measurable set $\cE \subset \cY$
such that for all $y \in \cE, U(Y) \leq  V(Y) - \epsilon$, and we have $P(\cE) > 0$.
Note that statements \eqref{eq:weaka} and \eqref{eq:weakb} are well-defined even if $U$ or $V$ are not measurable so that $U(Y)$ or $V(Y)$ are not random variables. Nevertheless, we shall abuse notation by abbreviating $U(Y)$ to $U$ and $V(Y)$, just like we do for random variables.  
Note that, if the events inside the probabilities below are measurable after all, then we have
\begin{equation}\label{eq:almost}
U\leq_{\theta} V \Leftrightarrow \forall P \in \cH(\theta): 
P(U \leq V) =1. 
\end{equation}
We also write $U >_{\theta} V$ if it is not the case that $U \leq_{\theta} V$; 
we write $U \geq_{\theta} V$ if it is not the case that $U <_{\theta} V$; 
and $U =_{\theta} V$ if $U \leq_{\theta} V$ and $U \geq_{\theta} V$; if $V$ and $U$ are measurable then clearly the corresponding analogues to \eqref{eq:almost} hold as well, e.g.
\begin{equation}\label{eq:almostb}
U<_{\theta} V \Leftrightarrow 
\forall P \in \cH(\theta): 
P(U \leq  V) =1 \text{\ and \ } \exists P \in \cH(\theta): P(U < V) > 0. 
\end{equation}
It is easily checked that $=_{\theta}$ establishes an equivalence relation on functions of $Y$, and relative to this relation, $
\leq_{\theta}$ is a partial order and $<_{\theta}$ is the corresponding strict order, i.e $U <_{\theta} V$ iff 
$U \leq_{\theta} V$ and not $U =_{\theta} V$. 
We shall freely use standard properties of this partial order (such as transitivity) below.

Moreover, we introduce the additional notation, for each function $B: \cY \rightarrow \cB$, where, in line with the  above, we abbreviate $B(Y)$ to $B$ and $\delta_{B(Y)}(Y)$ to $\delta_B$:
\begin{align*}
    \delta^{\circ}_B \leq_{\w} \delta_B  &   \Leftrightarrow \   
    \forall \theta \in \Theta: L_{B}(\theta,\delta^{\circ}_{B}) 
    \leq_{\theta} L_{B}\theta,\delta_{B}), \\
     \delta^{\circ} \leq_{\w} \delta  &  \Leftrightarrow \  
    \forall \theta \in \Theta, b \in \cB: L_b(\theta,\delta^{\circ}_b) \leq_{\theta} L_b(\theta,\delta_b),
\end{align*}
as such avoiding the cumbersome expression on the right whenever we can. 
Analogously, 
\begin{align*}
    \delta^{\circ}_B <_{\w} \delta_B &  \Leftrightarrow \ 
 \delta^{\circ}_B \leq_{\w} \delta_B \text{\ and\ }    \exists \theta \in \Theta: L_{B}(\theta,\delta^{\circ}_{B}) 
    <_{\theta} L_{B(Y)}(\theta,\delta_{B}), \\
     \delta^{\circ} <_{\w} \delta &  \Leftrightarrow \ 
\delta^{\circ} \leq_{\w} \delta \text{\ and\ }
    \exists \theta \in \Theta, b \in \cB: L_b(\theta,\delta^{\circ}_b) <_{\theta} L_b(\theta,\delta_b),
\end{align*}
and  correspondingly with $\geq_{\w}$ and $>_{\w}$. 
Finally, 
$\delta_B^{\circ} =_{\w} \delta_B$ 
is defined to be equivalent to `$\delta_B^{\circ} \geq_{\w} \delta_B$  and not $\delta_B^{\circ} >_{\w} \delta_B$'; similarly for `$\delta^{\circ} =_{\w} \delta$'.
\paragraph{Generalized Admissibility, Maximal Compatibility}
We can now generalize the definitions of   {\em Type-II strictly-better-than\/} and {\em admissibility\/} for general GNP decision problems in the main text: 
formally, we say decision rule $\delta$ is {\em Type-II strictly better\/} than $\delta^{\circ}$, simply if we have 
$$
\delta^{\circ} <_{\w} \delta. 
$$
As before, a decision rule $\delta^{\circ}$ is {\em admissible\/} if it is Type-I risk safe (according to the  definition underneath \eqref{eq:compatibleB} in the main text), and there is no other Type-I risk safe decision rule that is Type-II strictly better than $\delta$. 

We also  extend the definition of maximally compatible decision rule in the same way: formally, a {\em maximally compatible decision rule\/} relative to a given GNP decision problem and e-variable $S$ is any compatible decision rule $\delta^{\circ}$ which further satisfies that there is no other decision rule $\delta$ that is also compatible with $S$ and that is Type-II strictly better than $\delta^{\circ}$, with the extended definition of Type-II strictly better given above. 

We see that in the case of a GNP testing problem, whenever the events $L_b(\snull,\delta^{\circ}_b) >  L_b(\snull,\delta_b)$ are  measurable for all $b \in \cB$ (in particular whenever $\cY$ is countable), the probabilities in \eqref{eq:TypeIIbettera} and \eqref{eq:TypeIIbetterb} are well-defined and then the definition of strictly-better-than  coincides with the one given in the main text. But it continues to be valid in case we pick pathological $\cB$, $\{L_b: b \in \cB \}$ for which the events above are nonmeasurable --- which cannot be ruled out since we made no restrictions on the functions $L_b$, $\delta^{\circ}$ and $\delta$. Similarly, by replacing the definition \eqref{eq:TypeIRiskSafeSimplified} in the main text by \eqref{eq:TypeIRiskSafeSimplifiedB} we also make sure that Type-I risk safety is well-defined irrespective of whether $\sup_{b \in \cB} L_b(\snull,\delta_b(Y))$ is measurable or not (we could also have avoided such measurability issues using inner- and outer-measure \cite[Section 1.3]{Billingsley95} but this does not simplify the treatment so we decided against it).

In the same way, for a general GNP decision problem, the definition of strictly-better-than given here generalizes the one given in the main text  below \eqref{eq:TypeIrisksafeB} to the case where the events involved may be nonmeasurable; as a consequence, the definitions of admissibility for GNP testing and decision problems, and the definition of maximum compatibility relative to $S$ for GNP testing problems given in the main text, are all generalized by the definitions given here based on the generalized notion of strictly-better-than, 
and are valid irrespective of the measurability of the functions and events involved. 

Crucially for the proof of Lemma~\ref{lem:equalizer} below, the ordering relation $\leq_{\snull}$ is strong enough to imply inequality in expectation:
\begin{proposition}\label{prop:erc} Consider a GNP testing problem with null hypothesis $\cH(\snull)$ such that all $P \in \cH(\snull)$ are absolutely mutually continuous. 
Let $S= S(Y)$ and $S'= S'(Y)$ be nonnegative random variables such that for all  $P \in \cH(\snull)$, ${\bf E}_P[S]$ is finite. Suppose
$S \leq_{\snull} S'$. Then (a) for all $P \in \cH(\snull)$ we have  ${\bf E}_{P}[S] \leq  {\bf E}_{P}[S']$. Further, suppose 
$S <_{\snull} S'$. Then (b) for every $P \in \cH(\snull)$ we have  ${\bf E}_{P}[S] < {\bf E}_{P}[S']$.
\end{proposition}
\begin{proof}
(a) According to definition \eqref{eq:weaka}, for all $P \in \cH(\snull)$, $\epsilon > 0$ and every measurable set $\cE\subset \cY$ with for all $y \in \cE$, $S(y) > S'(y) +\epsilon$, we have 
$$
{\bf E}_{P}[S] =  {\bf E}_{P}[{\bf 1}_{Y \in \cE} \cdot S
+{\bf 1}_{Y \not \in \cE} \cdot S] 
=  {\bf E}_{P}[ {\bf 1}_{Y \not \in \cE} \cdot S] 
\leq {\bf E}_{P}[ {\bf 1}_{Y \not \in \cE} \cdot (S' + \epsilon)] 
= {\bf E}_{P}[S' + \epsilon]
$$
and the result follows. 

(b)
Fix $P \in \cH(\snull)$. According to  definition \eqref{eq:weaka}, for each $\epsilon > 0$ and measurable set $\cE$ with for all $y \in \cE$, $S(y) > S'(y) +\epsilon$, we have $P(\bar{\cE}) = 1$ (with $\bar{\cdot}$ denoting complement), whereas there exist $\delta > 0$ and $Q \in \cH(\snull)$ and measurable $\cF$ such that $S(y) < S'(y) - \delta$ on $\cF$ and $Q(\cF) > 0$. By mutual absolute continuity, we have $P(\cF) > 0$ as well, and therefore: 
\begin{align*}
{\bf E}_{P}[S] & = {\bf E}_{P}[{\bf 1}_{Y  \in \bar{\cE}} \cdot S]     
= {\bf E}_{P}[{\bf 1}_{Y \in \bar{\cE} \cap \cF} \cdot S
+ {\bf 1}_{Y \in \bar{\cE} \cap \bar{\cF}} \cdot S
]     \\
& \leq  {\bf E}_{P}[{\bf 1}_{Y \in \bar{\cE} \cap \cF} \cdot (S'- \delta)
+ {\bf 1}_{Y \in \bar{\cE} \cap \bar{\cF}} \cdot (S'+ \epsilon) 
\leq {\bf E}_{P}[S'] -  P(\cF) \delta + \epsilon.
\end{align*}
Since this holds for fixed $\delta > 0$ and for every $\epsilon > 0$, the result follows. 
\end{proof}
\subsection{General Form of Equalizing Maximal Compatibility Lemma}
\label{app:lemma}
Consider any GNP testing problem.
Fix any function $B: \cY \rightarrow \cB$. 
Generalizing \eqref{eq:equalizersimple}, we say that decision rule $\delta^{\circ}$ is {\em a.s. equalizing-maximally compatible\/} relative to e-variable $S$  {\em when restricted to $B$\/} if
\begin{equation}\label{eq:shoulder}
L_{B}(\snull,\delta^{\circ}_{B}) =_{\snull}  S, 
\end{equation}
where here and below, `a.s.' stands for `almost surely'. 
The following   lemma, generalizing Lemma~\ref{lem:equalizersimple}, is the key insight needed to prove Theorem~\ref{thm:maingeneral} below and hence its simplification Theorem~\ref{thm:mainsimple} in the main text. \begin{lemma}\label{lem:equalizer}
Fix any GNP testing problem. Suppose that $S$ is a sharp e-variable and  $\delta^{\circ}$ is a Type-I risk safe decision rule such that
there exists a function  $B: \cY \rightarrow \cB$ so that
$\delta^{\circ}$ is a.s. equalizing-maximally compatible  relative to $S$ when restricted to $B$, as in \eqref{eq:shoulder}.  Then $\delta^{\circ}$ is a.s. fully compatible with $S$, i.e. for  all $b \in \cB$, $L_b(\snull,\delta^{\circ}_b) \leq_{\snull}  S$. 
\end{lemma}
Just like in the simplified case with countable $\cY$, this immediately leads to a relevant corollary: 
\begin{corollary} 
Fix any GNP testing problem. If a decision rule $\delta^*$ is maximally compatible relative to a sharp e-variable $S$ and a.s. equalizing-maximally compatible when restricted to some function
$B: \cY \rightarrow \cB$ 
then any  $\delta^{\circ}$ that is a.s. equalizing-maximally compatible relative to $S$ when restricted to $B$ and Type-I risk safe must also be a.s. fully compatible with $S$, i.e. for  all $b \in \cB$, $L_b(\snull,\delta^{\circ}_b) \leq_{\snull}  L_b(\snull,\delta_b^*)$.
\end{corollary}
\begin{proof}{\bf \ [of Lemma~\ref{lem:equalizer}]}
By Proposition~\ref{prop:typeIsafety}, there must be some e-variable $S'$ such that $\delta^{\circ}$ is compatible with $S'$, i.e. 
\begin{equation}\label{eq:erc}
\text{for all $b\in \cB$, $y \in \cY$:}\     L_b(\snull,\delta^{\circ}_b(y)) \leq S'(y).
\end{equation}
By a.s. equalizing-maximal compatibility relative to $S$ when restricted to $B$, transitivity, and weakening \eqref{eq:erc}, we must therefore also have
\begin{equation}\label{eq:ercb1}
S =_{\snull} L_B(\snull,\delta^{\circ}_B) \leq_{\snull} S',
\end{equation}
so that $S \leq_{\snull} S'$.
Suppose by means of contradiction that, even stronger, $S <_{\snull}  S'$.
We know that for some $P_0 \in \cH(\snull)$,  ${\bf E}_{P_0}[S]=1$. But then Proposition~\ref{prop:erc}  gives that ${\bf E}_{P_0}[S'] > 1$, so $S'$ is not an e-variable and we have arrived at a contradiction. 
Since we already established $S \leq_{\snull} S'$  it follows that $S =_{\snull} S'$. But then using the inequality in   \eqref{eq:erc} 
gives for  all $b \in \cB$, $L_b(\snull,\delta^{\circ}_b) \leq_{\snull}  S$,  and the lemma is proved. 
\end{proof}

\subsection{Extension of Theorem~\ref{thm:mainsimple}  to general GNP decision problems}\label{app:finalresult}
\label{sec:typeIIproofs}
First, we extend the definitions of richness and sharpness  from GNP testing to decision problems in the obvious manner, by inserting `for all' quantifiers: 
we say that a GNP decision problem is {\em rich\/} relative to  e-collection $ \{S_{\theta}: \theta \in \Theta \}$ if for all $\theta \in \Theta$, the corresponding $\theta$-testing problem (as defined in the main text underneath Definition~\ref{def:gnpgeneral}) is rich relative to $S_{\theta}$. 
Relative to a given GNP decision problem, we say that e-collection 
$\{S_{\theta}: \theta \in \Theta \}$ is sharp if for all $\theta \in \Theta$, $S_{\theta}$ is sharp relative to the corresponding $\theta$-testing problem. 
\begin{theorem}\label{thm:maingeneral}
Consider any GNP decision problem. 
\begin{enumerate} 
\item 
Suppose that decision rule $\delta$ is admissible. Then there exists an e-collection $\ecol =  \{S_{\theta}: \theta \in \Theta \}$ 
such that $\delta$ is maximally compatible with $\ecol$.
\item Suppose that $\delta^*$ is a 
maximally compatible decision rule  relative to some 
e-collection  $\ecol= \{S_{\theta}: \theta \in \Theta \}$. 
If 
(a) all $P \in \cP$ are mutually absolutely continuous  and 
(b) $\ecol$ is sharp relative to the given GNP decision problem,  and 
(c) the GNP  decision problem is rich relative to $\ecol$, then $\delta^*$ is admissible. 
\end{enumerate}
\end{theorem}
\  \\
\begin{proof}{\bf \ [of Theorem~\ref{thm:maingeneral}]}

{\em Part 1}. 
Suppose that $\delta$ is admissible. Then $\delta$ is by definition Type-I risk safe. But then by the definition in the main text underneath \eqref{eq:compatibleB} there must be an e-collection $\ecol = \{S_{\theta}: \theta \in \Theta \}$ such that $\delta$ is  compatible with $\ecol$. Since $\delta$ is admissible, for every $\delta'$ with $\delta' >_{\w} \delta$ (i.e. $\delta'$ is strictly better than $\delta$) we have that $\delta'$ is not Type-I risk safe. But then $\delta'$ cannot be compatible with any e-collection, in particular it cannot be compatible with $\ecol$. Hence, there exist no $\delta'$ that is strictly better than $\delta$ and also compatible with $\ecol$; hence $\delta$ is maximally compatible. 

{\em Part 2}. 

Let $\delta^*$ be maximally compatible relative to $\ecol$  and let $\delta^{\circ}$ be another  Type-I risk safe decision rule. 
We will show that  $\delta^{\circ}$ cannot be Type-II strictly better than $\delta^*$; this implies the result. 

By our relative  richness assumption and the construction of $\delta^*$, we know that for all $\theta \in \Theta$,  there exists a function $B: \cY \rightarrow \cB$ with:
\begin{align}\label{eq:stra}
   \text{for all $y \in \cY$:}\  L_{B(y)}(\theta,\delta^*_{B(y)}(y)) = S_{\theta}(y).
\end{align}
We may assume that $\delta^\circ$ satisfies $\delta_{B}^{\circ} \geq_{\w} \delta_{B}^*$, otherwise we already know that it's not Type-II strictly-better. So, in particular, $L_B(\theta,\delta_{B}^{\circ}) \geq_{\theta} L_B(\theta,\delta_{B}^*)$ 
Now suppose by means of contradiction that $L_B(\theta,\delta_{B}^{\circ}) >_{\theta} L_B(\theta,\delta_{B}^*)$ for some $\theta \in \Theta$. 
Since  $\delta^{\circ}$ is Type-I risk safe, it must by definition also be compatible with an e-collection $\ecol' = \{S'_{\theta}: \theta \in \Theta \}$ so then we also have $S'_{\theta} >_{\theta} L_B(\theta,\delta_{B}^*)$. 
By Proposition~\ref{prop:erc}, using the assumption of mutual absolute continuity, we then have for the $P \in \cH(\theta)$ with ${\bf E}_{P}[S_{\theta}]=1$ (which must exist by sharpness) that 
$$ 
1=  {\bf E}_{P}[S_{\theta}] = {\bf E}_{P}[L_B(\theta,\delta^*_{B})]
< {\bf E}_{P}[S'_{\theta}]
$$
so $S'_{\theta}$ is not an e-variable and hence $\ecol'$ is not an e-collection, 
contradicting our assumptions (we note that all quantities inside the equation must be measurable, because $S_{\theta}$ an $S'_{\theta}$ are both e-variables, and hence measurable by definition). 
We may thus assume $L(\theta,\delta_{B}^{\circ}) =_{\theta} L(\theta,\delta_{B}^*)$.

The corollary of Lemma~\ref{lem:equalizer} above, applied with the corresponding $\theta$-GNP testing problem, 
now implies that for all $b \in \cB$ (hence not just for $B$!) we have $L_b(\theta, \delta^{\circ}_b) \leq_{\theta} L_b(\theta, \delta^*_b)$. Since we can  make this argument for all $\theta \in \Theta$, it follows that for all $b \in \cB$,  $\delta_{b}^{\circ} \leq_{\w} \delta_{b}^*$. Therefore $\delta^{\circ}$ is not Type-II strictly better than $\delta^*$; the theorem is proved. 
\end{proof}

\section{Supporting Information for Section~\ref{sec:eposterior}.\ref{sec:obayes}}
\paragraph{Proof for Claim underneath \eqref{eq:therealmccoy}}
Fix some $\epsilon > 0$.
For simplicity we fix $\theta^*=0$ and $n=1$ (so that $Y=X_1 = \hat\theta(X_1)$;) extension of the following argument to general sampling distributions $\theta^*$ and $n > 1$ is straightforward (for $\theta^* \neq 0$, use $Y' = Y- \theta^*$; for $n >1$, simply adjust the variance). 

We will construct $B(y)$ such that if $y =  \epsilon$, the CI  corresponding to $B(y)$ will be a single point at $\epsilon$; if $y$ gets larger, the CI widens but no matter how large $y$, it will never cover the true $\theta^*=0$. 
To this end, fix any strictly positive, strictly decreasing function $g_0$ with $g_0(\epsilon) = \epsilon$. 
We will take $B(y)$  such that $\delta_{B(Y)}(Y)$  has as its left-end $\Ltheta = g_{0}(y) = y - h(y)$  where $h(y) := y- g_0(y)$. The CI being by definition symmetric around $y$, we must then have  $\Rtheta  = y + h(y) = 2y - g_0(y)$.
Since for any $0 \leq \alpha \leq 1$, the $(1-\alpha)$-CI coincides with the $(1-\alpha)$ credible interval taken symmetrically around the MLE $\hat\theta = y$,  
both its left- and right-tail must have posterior weight $\alpha/2$. Since the posterior has a Gaussian density with mean $y$ and variance $1$, we must thus have  $\alpha/2 = \int_{- \infty}^{g_{0}(y)} f_{y}(u) d u$, where we by denote $f_{\mu}$ the density of a normal with variance $1$ and mean $\mu$, so $\alpha = 2 F_y(g_0(y)) = 2 F_0(- y+ g_0(y))$, where $F_{\mu}$ is the CDF of a normal with mean $\mu$ and variance $1$. It follows that $B(y)$ must be equal to $1/\alpha = 1 / (2 F_0(-y + g_0(y)))$

If data are actually sampled from $\theta^*=0$, 
then the expected loss we {\em actually\/} make can be calculated in steps as follows: 
\begin{align*}
& {\bf E}_{Y \sim P_{\theta^*}}[L_{B(Y)}(\theta^*,\delta_{B(Y)}(Y))] = 
{\bf E}_{Y \sim P_{0}}[B(Y) \cdot {\bf 1}_{0 \not \in \delta_{(B)}(Y)}]
\geq {\bf E}_{Y \sim P_{0}}[B(Y) \cdot {\bf 1}_{Y \geq \epsilon}]  \\
& =  \int_{\epsilon}^{\infty} f_{0}(y) \cdot \frac{1}{2 \cdot F_0(g_0(y) -y)) } d y \\
& \geq \frac{1}{2} \cdot 
\int_{\epsilon}^{\infty} \exp\left(- \frac{y^2}{2} \right) \cdot
\exp\left(\frac{(y- g_0(y))^2}{2}\right) (y - g_0(y)) d y
\\ & \geq  \sqrt{\frac{\pi}{2}} \cdot 
\int_{\epsilon}^{\infty} \exp\left(-  y g_0(y) \right) \cdot
(y - g_0(y)) d y,
\end{align*}
where we used the standard result that, with $P_0$ denoting a standard normal distribution,  $P_0(Y \geq c) \leq \exp(-c^2/2)/(c \cdot \sqrt{2 \pi})$. 
Clearly the integral diverges for many choices of $g_0$ satisfying our requirements; for example, we can take $g_0(y) = \epsilon^2/y$ (which works for all $\epsilon > 0$) or (if we want to make the probability of large $B$ smaller) we can set $g_0(y) = \epsilon \cdot (\log (y+ \exp(\epsilon) - \epsilon))/y $ if $\epsilon$ is set to $2$; then $\exp(-y g_0(y)) = (y + \exp(2) - 2)^{-2}$.  
In the table in the main text we took the former choice to see how a typical sample of the $B$'s, and corresponding $\alpha$'s and CI's might look like. 

\paragraph{Proof that the average in \eqref{eq:inductivebehaviourci} may converge to infinity with probability 1} 
Assume a sequence of independent studies $Y_{(1)}, Y_{(2)}, \ldots$, all of which are of the form $Y_{(j)}= (X_{(j),1})$ and thus have sample size 1 (we can treat larger sample sizes, and simple sizes varying from study to study, by  thinking of the $Y_{(j)}$ as $z$-scores summarizing studies of varying sample size). Instantiate $\epsilon = 0.01$ and set $B_{(j)} := 1/(2 F_0(-Y_{(j)} + \epsilon^2/Y_{(j)}))$ as above if $Y_{(j)} \geq \epsilon$ and $B_{(j)}= 1$ otherwise. Suppose that the $Y_{(j)}$ are all independently sampled from the same $\theta^*= 0$. Here is a sample (generated i.i.d. by {\tt R}) of 20 corresponding $B_{(j)}$ 
(recall that for each $j$, the corresponding  produced interval $\delta_{B_{(j)}}(Y_{(j)})$ is equal to the standard $(1- \alpha_{(j)})$-confidence interval, with $\alpha_{(j)} = 1/B_{(j)}$): 
\begin{align}\label{eq:sample}
& 1.15 ,1 , 3.44 , 1.09,  1.91,  4.17, 10.40,  1.11  ,1 ,1 ,1  \nonumber \\
& 1.47,  1.31, 1 , 1 , 1,  2.28,  1.76  ,1,1,1
\end{align}
While the sequence looks rather innocuous, using \eqref{eq:therealmccoy}, with $A$ in $\hat\theta \pm A$ chosen by \eqref{eq:below}, we see the limit in \eqref{eq:inductivebehaviourci} will go a.s. to $\infty$ rather than to $1$. The example was deliberately designed to give an extreme discrepancy --- in more realistic examples, the difference will presumably not be infinite but without knowing the dependency between $Y$ and $B$ there is no way to assess it. 

\subsection*{Proof for Example~\ref{ex:prebexample}}\label{app:eposterior}
We first treat e-variables corresponding to one-sided tests, for which we can give exact results. 
To this end, let  $\meanmin < \mean < \meanplus$  be 
\begin{equation} \label{eq:nstarfunb}
  \frac{1}{2} 
n^* (\mean - \meanplus)^2 = \frac{1}{2}
n^* (\mean - \meanmin)^2 =
\log \frac{1}{\alpha^*}.
\end{equation}
(note that $\log (2/\alpha^*)$ in \eqref{eq:nstarfunb} in the main text has been replaced by $\log(1/\alpha^*)$ here). 

The  {\em uniformly most powerful Bayes factor\/} \cite{johnson2013uniformlyB} for a 1-sided test at sample size $n^*$ and level $\alpha^*$ of $\Hnull = \{P_{\mean}\}$ vs. $\Halt = \{P_{\mean'}: \mean' > \mean\}$ 
(or, $\Halt = \{P_{\mean'}: \mean' < \mean\}$  respectively) is given by $S^+_{\theta} := \frac{p_{\meanplus}(y)}{p_{\mean}(y)}$ (respectively, $S^ -_{\theta} := \frac{p_{\meanmin}(y)}{p_{\mean}(y)}$).
 Straightforward rewriting now gives:  
\begin{align}
S^+_{\theta}(y) & = 
 \frac{p_{\meanplus}(y)}{p_{\mean}(y)} = 
 \frac{e^{-n 
\log \frac{p_{\hat\theta}(y)}{p_\meanplus(y)}}}{
e^{-n \log \frac{p_{\hat\theta}(y)}{p_\theta(y)}}}
  =   \frac{e^{-(n/2) (\hat\mean - \mean - U)^2}}
{e^{- (n/2) (\hat\theta - \theta)^2}}
 = 
\frac{e^{n \cdot  (\hat\mean - \mean) U }}{ e^{n U^2/2}
} 
\label{eq:onesided}
=  e^{-n U^2/2 +n (\hat\theta-\theta) U}
\end{align}
where 
$$U=\sqrt{ \frac{2  (\log (1/\alpha^*) }{ n^*}}= \sqrt{ \frac{2 \cdot c \cdot  \log (1/\alpha) }{ n}} \text{\ 
with\ } c= 
\frac{n^*}{n} \cdot \frac{\log (1/\alpha)}{\log(1/\alpha^*)}
.$$
We see that $S^+_{\theta}$ is strictly increasing in $\hat\theta$, so it is  $\geq 1/\alpha$ iff $\hat\theta \geq \theta_R$, where $\theta_R$ is the solution to
$$
 e^{-n U^2/2 +n (\theta_R-\theta) U} = \frac{1}{\alpha}. 
$$
Straightforward calculation shows that this is the case iff $\theta_R- \theta$ is equal to
\begin{align}\label{eq:border}
 \sqrt{\frac{2}{n} \cdot  \log(1/\alpha)} \cdot g(c)
\text{\ with\ } g(c) = \frac{c+1}{2\sqrt{c}} = \frac{1}{2} \left( c^{1/2} + c^{-1/2} \right).
\end{align}
An analogous calculation gives that $S^{-}_{\theta}$ is decreasing in $\hat\theta$ and $\geq 1/\alpha$ iff $\hat\theta \leq \theta_L$, with $\theta- \theta_L$ equal to \eqref{eq:border}. 

For the two-sided e-variable $S_{\theta} = (1/2) S^+_{\theta} + (1/2) S^-_{\theta}$, a sufficient condition for $S_{\theta} \geq 1/\alpha$ is then that
$$
S^+_{\theta} \geq \frac{2}{\alpha} \text{\  or\  } S^-_{\theta} \geq \frac{2}{\alpha}.
$$
Therefore, if we apply the above with $\alpha^{*'}$ set to $\alpha^*/2$ and $\alpha'$ set to $\alpha/2$, we get that $\theta^+, \theta^-$, $S^+_{\theta}, S^-_{\theta}$, $S_{\theta}$ and $c$ are now defined as in the main text, and a sufficient condition for $S_{\theta} \geq 1/\alpha$ is that 
\begin{equation}
    \label{eq:TheBound}
    |\hat\theta - \theta | \geq \sqrt{\frac{2}{n} \cdot  \log \frac{2}{\alpha}} \cdot g(c),
\end{equation}
which was our claim in the main text. 
Since for fixed $\alpha$,  for all but the smallest $n$, whenever $S^+_{\theta} \geq \frac{2}{\alpha}$, we have that $S^-_{\theta}$ must be very close to $0$, since it decreases exponentially in $n$ (and the same with $S^+_{\theta}$ and $S^-_{\theta}$ interchanged), we find that \eqref{eq:TheBound} is quite tight in practice.

\commentout{
Fix some GNP testing problem. When assessing whether a decision rule $\delta$ is compatible or maximally compatible with some e-collection $\ecol= \{S_{\theta}: \theta \in \Theta \}$, a pre-condition for this to have any chance of being the case is that for all $\theta \in \Theta$, all $b \in \cB$, we have $\inf_{a \in \cA_b} L_b(\theta,a) \leq \inf_{y \in \cY} S(y)$. If this is the case we say that the GNP testing problem allows for {\em potential compatibility} with $S$. For simple GNP testing problems, we have $L_b(\snull,0) = 0$ for all $b \in \cB$. Clearly, under that assumption, potential compatibility automatically holds.}

\section{Supporting Information for Section~\ref{sec:discussion}}
\paragraph{On the Type-I Risk Upper Bound $\ell$}
\commentout{
\paragraph{Where does the Type-I risk bound $\ell$ come from?}
Whereas $B$ may arbitrarily depend on data, the upper bound $\ell$ has to be set independently of the data after all. 
It may still vary from decision problem to decision problem though (in the SI Appendix we explain what this means in terms of Neyman's inductive behavior paradigm). In many practical testing problems, we might expect that for all $b \in \cB$, we have that $\cA_b$, contains a special action $0$, which stands for `do nothing' (keep the status quo),  which would then have the same Type-II loss under all $b \in \cB$, i.e. there is an $\ell'$ such that for all $b \in \cA_b$, $L_b(\alt,0)= \ell'$. We might then simply set $\ell = \ell'$, making sure  we expect our result (will all costs and benefits incorporated), whatever action we take, to be not worse than the result of doing nothing. We note that the developments in this paper work for arbitrary such $\ell$; setting $\ell=1$ in Section~\ref{sec:math} was only done for convenience (see the SI Appendix). 
}
Here we discuss why normalizing $\ell$ to 1 in \eqref{eq:TypeIRiskSafeSimplified} and \eqref{eq:TypeIRiskSafeSimplifiedB} is not harmful, and what we mean when we say (as we did in the discussion Section~\ref{sec:discussion}) that `$\ell$ can be chosen differently from problem to problem, but it needs to be chosen independently of the data observed in that problem'.

Suppose you are a statistician, performing hypotheses tests within a variety of domains. Let $Y_{(1)}, Y_{(2)}, \ldots$ be the sequence of samples, taking values in potentially different sets $\cY_{(1)}, \cY_{(2)}, \ldots,$ and associated with different null hypotheses $\cH_{(0)}(\snull), \cH_{(1)}(\snull),\ldots$ and  associated GNP testing problems, that you are confronted with in your professional career. 
We will assume the $Y_{(j)}$ are all independent. Each time $j$ that you perform a hypothesis test, policy makers provide you with an upper bound $\ell_{(j)}$ (e.g. set equal to $L(\alt,0)$, the cost of maintaining the status quo and doing nothing, see Section~\ref{sec:discussion}), that may depend on previous outcomes but, given $Y_{(1)}, \ldots, Y_{(j-1)},$ must be independent of $Y_{(j)}$. You also are given the loss function $L_{B_{(j)}}$ with associated action space $\cA_{B_{(j)}}$. You change this into loss function $L'_{B_{(j)}} := L_{B_{(j)}}/\ell_{(j)}$, and you advise an action by applying a maximally compatible decision rule $\delta_{(j)}$ relative to $L'_{B_{(j)}}$. 
By multiplying both sides in the definition of Type-I risk safety (\eqref{eq:TypeIRiskSafeSimplified} or \eqref{eq:TypeIRiskSafeSimplifiedB}) with $\ell_{(j)}$ again, we see that Theorem~\ref{thm:mainsimple} implies that at each time $j$, your Type-I risk is bounded by $\ell_{(j)}$, i.e. $\sup_{P \in \cH_{(j)}(\snull)} {\bf E}_{P}[L_{B_{(j)}}(\snull,\delta_{B_{(j)}}(Y_{(j)})] \leq \ell_{(j)}$. 

Now, suppose that an outside evaluating agency is interested in your performance whenever the imposed bound is close to some specific $\ell^*$. Thus, after you have engaged in $m$ hypothesis testing problems, they look at the subset $\cI_{m,\delta} := \{ j \in [m]: |\ell_{(j)} - \ell^* | \leq \delta \}$ for some small $\delta > 0$. 
Now, let us assume that the process determining the $\ell_{(j)}$s is such that $\lim_{m \rightarrow \infty} | \cI_{m,\delta}| = \infty$, almost surely, i.e.  a risk bound close to $\ell^*$  will eventually be chosen infinitely often. Then, by the strong law of large numbers, we also have that 
$$
 \lim \sup_{m \rightarrow \infty} \frac{1}{| \cI_{m,\delta} |} \sum_{j \in \cI_{m}} L_{B_{(j)}}(\snull,\delta_{B_{(j)}}(Y_{(j)})] \leq \ell^*+ \delta. 
$$
Thus, your e-value based statistical hypothesis tests have a Neymanian inductive behavior interpretation: as long as the bounds $\ell_{(j)}$ themselves do not depend on data $Y_{(j)}$, then in the long run, among all tests in which the imposed bound was within $\delta$ of  $\ell^*$, you will achieve average loss that is also within $\delta$ of $\ell^*$. In particular the normalization to $\ell_{(j)}=1$ in the definitions in Section~\ref{sec:math} does not affect this guarantee. 

\subsection{Evidential Interpretation of E-Values}\label{sec:evidence}
Most practitioners still interpret p-values in a Fisherian way, as a notion of evidence against the null. Although this interpretation has always been highly controversial, it is to some extent, and with caveats (such as `single isolated small p-value does not give substantial evidence' \cite{mayo2018statistical} or `only work with special, {\em evidential\/} p-values \cite{Greenland22}'), adopted by highly accomplished statisticians, including the late Sir David Cox \cite{Cox18,mayo2006frequentist}. Even Neyman \cite{Neyman76} has written `my own preferred substitute for `do not reject $H$' is `no evidence against $H$ is found'. In light of the results of this paper, one may ask if, perhaps, {\em e-values are more suitable than p-values\/} as such a measure. Although a proper analysis of such a claim warrants (at the very least) a separate paper, we briefly make the case here. At first sight this question may seem orthogonal to the Neymanian `inductive behavior'  stance adopted in this paper--- as has often been noted \cite{BergerS87,HubbardB03,Berger03,Hubbard04}, Fisher's and Neyman's interpretations of testing seem incompatible.
Nevertheless (echoing a point made by error statisticians \cite{mayo2018statistical} and likelihoodists \cite{royall1997statistical} alike), for any notion of `evidence the data provide about a hypothesis $\Hnull$' to be meaningful at all, there have to be circumstances, perhaps idealized, in which additional knowledge $\textsc{k}$ is available, and together with {\sc k}, the evidence can be operationalized into reliable decisions (for, if there were no such circumstances, obtaining `high evidence' for or against a claim could never have any empirical meaning whatsoever...). 
For the likelihoodists's notion of evidence  \cite{royall1997statistical,edwards1984likelihood}, i.e. a likelihood ratio between simple $\Hnull$ and $\Halt$, this additional knowledge {\sc k} would be a trustworthy prior probability on $\{\Hnull,\Halt\}$ --- once this is supplied, a DM can use Bayes'  theorem to come up with a posterior which can then lead to optimal decisions against arbitrary loss functions. 
For the notion of evidence against $\Hnull$ as a p-value, this {\sc k}  would comprise a guarantee that a specific, a priori fixed and known sampling plan,  would have been followed (otherwise the p-value would be undefined), and an  a priori specified $\alpha$, and knowledge that the decision would be of the simple form `accept'/`reject'. This {\sc k}, however, is additional knowledge of a {\em very\/} specific kind (essentially, what we called the BIND assumption). In other situations, it is not clear at all how to operationalize evidence-by-p-value into decisions. Now, if we accept e-values as evidence against the null, the set of circumstances under which we can operationalize the evidence is much wider, as shown in this paper. Having thus  direct empirical content in a wider variety of situations,  $e$ would seem preferable over $\pval$ (a). 

Note that I am not saying that evidence should {\em invariably\/} be a `stepping-stone' towards a decision\footnote{Thanks to a referee for prompting this important clarification.}; evidence seems a more general notion than that. I am only saying that if there are broad sets of circumstances in which it {\em is\/} a stepping stone, this may be a good rather than a bad thing.  

Add to this: (b) if $\Hnull$ and $\Halt$ are simple, the e-value {\em coincides} with the likelihood ratio, i.e. the main competing notion of evidence; (c) if $\Hnull$ is simple yet $\Halt$ is not, a special type of e-value coincides with a recently proposed Bayesian notion of evidence (the {\em support interval\/} \cite{PawelLW22,wagenmakers2020support}); (d) unlike Bayesian methods, e-values can be constructed even if no clear alternative can be formulated and if the setting is highly nonparametric; and (e) in contrast to p-values, e-values remain meaningful if some details of the sampling plan are unknown or unknowable and if information from several interdependent studies is combined \cite{GrunwaldHK19B,ramdas2023savi}. 
In fact, this may be the most important observation: if a scientific study is performed,  and, {\em because\/} the scientific study seemed promising, a second study was performed, then we would lie to report the evidence against the null provided by both studies taken together. Yet, while for e-values this is no problem (we can multiply the e-values of the individual studies), it is next to impossible to calculate a valid p-value for the two studies taken together  --- this is the main point of \cite{GrunwaldHK19B}. The fact that they cannot be calculated in such a standard scenario would seem to make them unsuitable as a notion of evidence. If we tae (a)---(d) together though, the case for e-values as evidence seems strong. 

A similar comment pertains to Mayo's {\em  error statistics\/} philosophy with its concept of {\em severe testing\/} \cite{mayo2006frequentist,mayo2018statistical,Yanofsky19}: currently, Mayo's notion of severity is, at least in simple cases, indirectly based on p-values \cite[page 144]{mayo2018statistical}. In light of the above, it might be preferable to use e-values instead.

\end{document}